\documentclass[a4paper]{article}
\usepackage[totalwidth=480pt, totalheight=680pt]{geometry}

\usepackage{lmodern}

\usepackage{soul}
\usepackage{url}
\usepackage[hidelinks]{hyperref}
\usepackage[utf8]{inputenc}
\usepackage[small]{caption}
\usepackage{amsmath}
\usepackage{amsthm}
\usepackage{graphicx}
\usepackage{booktabs}
\usepackage{algorithm}
\usepackage{algorithmic}
\usepackage[switch]{lineno}

\title{Efficient Algorithms for Electing Successive Committees}

\usepackage{authblk}
\author[1]{Pallavi Jain}
\author[2]{Andrzej Kaczmarczyk}
\date{}
\affil[1]{\small Indian Institute of Technology, Jodhpur, India}
\affil[2]{\small Department of Computer Science, The University of Chicago, Chicago,
United States}
\affil[ ]{\small \textit{pallavi@iitj.ac.in, akaczmarczyk@uchicago.edu}}

\usepackage{wrapfig}

\newcommand{\appsymb}{$\bigstar$}

\usepackage{xspace}
\newcommand{\orderedlistingof}[2]{\ensuremath{#1_1, #1_2, \ldots, #1_#2}}
\newcommand{\orderedsetof}[2]{\ensuremath{\{\orderedlistingof{#1}{#2}\}}}

\newcommand{\namedorderedsetof}[3]{\ensuremath{{#1}=\orderedsetof{#2}{#3}}}

\newcommand{\wsetpacking}{\textsc{Weighted Set Packing}}

\newcommand{\yes}{\textsf{\small yes}}
\newcommand{\no}{\textsf{\small no}}
\newcommand{\true}{\textsf{\small true}}
\newcommand{\false}{\textsf{\small false}}
\newcommand{\failure}{\textsf{\small failure}}

\newcommand{\pnamesuffixfull}[2]{\textsc{Succesive Committees Election}}

\newcommand{\pname}[2]{\textsc{\ensuremath{#1}-\ensuremath{#2}-SCE}}
\newcommand{\pgeneralnamefull}{\textsc{$\alpha$-$\beta$-\pnamesuffixfull}}
\newcommand{\pgeneralname}{\textsc{$\alpha$-$\beta$-SCE}}

\newcommand{\np}{\ensuremath{\mathrm{NP}}}

\newcommand{\fpt}{\ensuremath{\mathrm{FPT}}\xspace}

\newcommand{\wtwo}{\ensuremath{\mathrm{W[2]}}}                                   
\newcommand{\w}{\ensuremath{\mathrm{W}}}

\newcommand{\OO}{\mathcal{O}}

\newcommand{\nphard}{{\np-hard}\xspace}

\newcommand{\naturals}{\mathbb{N}}

\usepackage{amssymb}
\usepackage{mathtools,thmtools}
\usepackage{nicefrac}

\usepackage{multirow}
\usepackage{array}

\usepackage{xcolor}
\definecolor{darkgreen}{rgb}{0,0.5,0}
\hypersetup{
    unicode=false,          %
    colorlinks=true,        %
    linkcolor=purple,          %
    citecolor=darkgreen,    %
    filecolor=magenta,      %
    urlcolor=blue           %
}

\usepackage{tikz}
\usetikzlibrary{positioning,fit,matrix}
\usepackage{thmtools,thm-restate}
 
\usepackage{natbib}
\usepackage[nameinlink]{cleveref}

\newtheorem{theorem}{Theorem}

\newtheorem{corollary}{Corollary}
\newtheorem{lemma}{Lemma}
\newtheorem{proposition}{Proposition}

\newtheorem{definition}{Definition}
\newenvironment{proofsketch}{%
  \proof}{\endproof}

\DeclareMathOperator{\score}{sc}
\DeclareMathOperator{\egal}{egal}
\DeclareMathOperator{\util}{util}
\DeclareMathOperator{\pos}{pos}

\DeclareMathOperator{\beg}{beg}

\DeclareMathOperator{\head}{head}
\DeclareMathOperator{\tail}{tail}

\DeclareMathOperator{\cc}{CC}

\DeclareMathOperator{\approvalScore}{App}
\DeclareMathOperator{\egalCCScoreA}{trCC}
\DeclareMathOperator{\egalCCScoreO}{eCC}

\DeclareMathOperator{\coverageScore}{AppCC}

\DeclareMathOperator{\utilOp}{util}
\DeclareMathOperator{\egalOp}{egal}

\newcommand{\WSFamily}{\ensuremath{\mathcal{F}_{\text{WS}}}}

\newcommand{\Co}[1]{\mathcal{#1}}

\DeclareMathOperator{\reqqual}{rq}

\newcommand{\colorful}{{{\color{red}c}{\color{blue}o}{\color{olive}l}{\color{violet}o}r{\color{orange}f}{\color{purple}u}{\color{teal}l}}\xspace}
\newcommand{\candidates}{\ensuremath{C}}
\newcommand{\candidatesnr}{\ensuremath{m}}
\newcommand{\voters}{\ensuremath{V}}
\newcommand{\votersnr}{\ensuremath{n}}
\newcommand{\approvals}{\ensuremath{A}}
\newcommand{\ranking}{\ensuremath{\succ}}
\newcommand{\series}{\ensuremath{\mathcal{W}}}

\newcommand{\timelimit}{\ensuremath{\tau}}

\newcommand{\vallowerbound}{\ensuremath{\eta}}
\newcommand{\commsize}{\ensuremath{k}}
\newcommand{\candquota}{\ensuremath{f}}

\newcommand{\committeesSeries}{\ensuremath{\mathcal{S}}}

\newcommand{\PSdivision}{\ensuremath{\mathcal{D}}}
\newcommand{\PSdivsets}{\ensuremath{\mathcal{S}}}
\newcommand{\PSdivfil}{\ensuremath{\mathcal{F}}}
\newcommand{\PSdivisionSize}{\ensuremath{d}}
\newcommand{\PSorder}{\ensuremath{\rho}}

\usepackage[textsize=scriptsize]{todonotes}

\usepackage[inline]{enumitem}

\newenvironment{problemQuote}%
  {\list{}{\leftmargin=0.04in\rightmargin=0.00in}\item[]}%
  {\endlist}
\newcommand{\probDef}[3]{%
 \begin{problemQuote}
  #1\\
  \textbf{Input:} #2\\
  \textbf{Question:} #3%
 \end{problemQuote}
}

\DeclareRobustCommand{\abbrevcrefs}{%
  \Crefname{theorem}{Thm.}{Thms.}%
  \Crefname{example}{Ex.}{Exs.}%
  \Crefname{proposition}{Pr.}{Prs.}%
  \Crefname{corollary}{Cor.}{Cors.}%
}

\DeclareRobustCommand{\Shcref}[1]{{\abbrevcrefs\Cref{#1}}}

\begin{document}
\maketitle 

\begin{abstract}
	In a recently introduced model of \emph{successive committee elections}
	(Bredereck et al., AAAI-20) for a given set of ordinal or approval preferences
	one aims to find a sequence of a given length of ``best'' same-size
	committees such that each candidate is a member of a limited number of
	consecutive committees. However, the practical usability of this model remains
	limited, as the described task turns out to be NP-hard for most selection
	criteria already for seeking committees of size three. Non-trivial or somewhat
	efficient algorithms for these cases are lacking too. Motivated by a desire to
	unlock the full potential of the described temporal model of committee
	elections, we devise (parameterized)
  algorithms that effectively solve the mentioned hard cases in realistic
  scenarios of a moderate number of candidates or of a limited time horizon.
\end{abstract}

\section{Introduction}

A non-profit organization (NPO) offers a 3-day personal development workshop for
teenagers in a remote location. Featured activities
include discussion sessions 
with three expert counselors who
tackle the participants' questions in various topics of developing self
awareness.
Potential
counselors
agreed to participate in 
such 
discussions for at most two consecutive days
to avoid excessive traveling and fatigue. 
Every counselor specializes in a limited selection of topics regarding self
awareness.
The NPO wants to select three groups of three counselors, one group per day, to
offer the participants as broad experience as possible. Hence, the selected
groups must obey the counselors' consecutiveness requirement and also
guarantee a diverse selection of topics covered each day by the respective
counselor pair.

The NPO's task seemingly can be modeled as a~\emph{multiwinner voting}
task~\citep{fal-sko-sli-tal:b:multiwinner-voting,lac-sko:b:multiwinner-approval}.
Identifying counselors as candidates and each subarea as a voter, who approves
for the candidates representing counselor's expertise, we want to select a
diverse group of candidates. However, the classical model neglects the temporal
aspect of the problem and thus does not align well with the task. Precisely, it
fails to group the counselors in teams of three to serve on the three days of
the workshop while adhering to the consecutiveness requirement.

The described shortcoming has recently been
addressed by~\citet{bre-kac-nie:c:successive-committees}, who proposed a
suitable framework of~\emph{successive committee elections}. Here, based on a
collection of votes, one selects a collection of multiple, ordered,
same-size groups, called \emph{committees}, of ``best'' candidates.
In line with our toy example, each candidate in the selected committees must be
a member of a single contiguous block of at most a given number of committees.
Besides introducing the new model, the authors have studied associated problems
through algorithmic lens.
They studied four successive committee rules focused on maintaining the
diversity of the chosen committees based on the Chamberlin--Courant
rule~\citep{CC83}. Further, they considered extensions of the widespread
approval voting rule and weakly-separable scoring rules.
Their study identified cases solvable in polynomial time
mostly related to finding series of committees that consist of
two candidates. While they showed that with very few exceptions the same task
for committees of size three or more is computationally hard, they abstained
from providing algorithms dealing with such cases. 

Motivated by the desire of providing means of computing successive committees
for criteria considered by~\citet{bre-kac-nie:c:successive-committees},
we sidestep their hardness results
by applying advanced algorithmic methods from parameterized complexity theory.
We analyze the same set of rules they did and contribute several positive
algorithmic results, as shown in~\Cref{tab:results}. Consequently,
\begin{enumerate*}[label={\small(\arabic*)}, itemjoin={{, }}, itemjoin*={{,
	and }}]
  \item we provide the first algorithms that deal with the hard scenarios for
    committee size bigger than two, including cases focusing on diverse
    committees similar to our introductory example
  \item our methods are furthermore computationally efficient (in the
    parameterized sense) for cases with a small number of candidates or
    committees to be selected.
\end{enumerate*}
The mentioned cases of interest seem quite likely in practice. Indeed, planning
too far ahead usually bears unacceptable risks (like unpredictable changes of
voters preferences or dropping out of candidates). Too big pool of candidates is
often undesired due to human perception limitation and is usually avoided by
shortlisting (in a broad sense, e.g.,\ by requiring petition signatures,
deliberation processes, organizing pre-selection). 
As our algorithmic contributions extend the applicability of successive
committee elections, we directly respond to the call
of~\citet{boe-nie:c:broadening-agenda-comsoc} to consider different paths of
tractability of new models\footnote{Their taxonomy classifies successive
committee elections as the setting of  \textbf{O}rdered \textbf{O}ne profile
\textbf{M}ultiple solutions (O-OM).} better capturing \emph{the changing nature
of real-world problems}.

Some other models
incorporate the time aspect into classical committee elections additionally
allowing votes to change their vote over time.
\citet{bre-flu-kac:c:votes-change-committees-not} introduce such a model and
analyze the (parameterized) computational complexity of several related questions.
Their results, complemented with those
of~\citet{kel-ren-zsc:c:parameterized-diverse-multistage}, offer a comprehensive
computational landscape.
Importantly, their model does not generalize the one we
consider, as their requirements for the outcome on a series of committees are
focused on (dis)similarity of neighboring committees.
\citet{del-flu-bre:c:egalitarian-equitable-sequences-of-committees} deepen the study of the model of \citet{bre-flu-kac:c:votes-change-committees-not} 
by studying new questions related mostly
to the fairness of the outcomes towards the voters.
The
dynamic nature of preferences considered
by the three listed works
makes approaches therein
impossible
to adapt
to our problems.
In the light of the mentioned literature, our work takes a significant step
towards addressing the somewhat neglected algorithmic study of successive
committee elections.

Another related scenario of time-dependent voting where the task is to
select a single candidate for each time step can be seen as a series of
committees of size one. The literature on this model mostly focuses on
proportionality and fairness from
axiomatic~\citep{lac:c:perpetual-voting,lac-mal:x:perpetual-axiomatic-lens,cha-goe-pet:c:proportional-aggregation-for-sequential-dm},
algorithmic~\citep{elk-obr-teh:c:proportionality-in-temporal-voting,bul-haz-pag-ros-tal:a:jr-for-perpetual}
and experimental
perspectives~\citep{bul-haz-pag-ros-tal:a:jr-for-perpetual,cha-goe-pet:c:proportional-aggregation-for-sequential-dm}.
See a survey by~\citet{elk-obr-teh:c:temporal-fairness} for more details
on this topic.

\section{Preliminaries}\label{sec:prelim}
 For a positive integer~$x$, we use $[x]$ to denote set $\{1, 2, \ldots, x\}$.
 An \emph{election}~$E = (\candidates, \voters)$ consists of
 $\candidatesnr$~candidates~$\candidates = \{c_1, c_2, \ldots,
 c_\candidatesnr\}$ and a collection~$\voters = \{v_1, v_2, \ldots,
 v_\votersnr\}$ of~\votersnr{}~voters. We study \emph{approval} and
 \emph{ordinal} preferences. In the former preference type, we associate a
 voter~$v_i \in \voters$ with their \emph{approval set}~$\approvals(i)$ of
 candidates that $v_i$ approves. In the ordinal preferences model each voter~$v_i$
 ranks all candidates and so is identified with their (total and strict)
 \emph{preference order}~$\ranking_i$. We denote by~$\pos_i(c)$ a position of
 candidate~$c \in \candidates$ in some voter~$v_i$'s ranking~$\ranking_i$.
 There are multiple mathematical ways of relating ordinal preferences to
 approval preferences.\footnote{For example: approvals are either complete
 rankings with ties or can be constructed from rankings by letting each voter
 approve a number of their top candidates.} However, they all require
 decisions somewhat arbitrary from the practical perspective.
 
 \paragraph{Fixed-Parameter Tractability.}
 We say a computational problem is fixed-parameter tractable for some
 parameter~$x$ (being a part of the input) if there is an (parameterized) algorithm solving
 every instance~$\mathcal{I}$ in time~$\OO(f(x)|\mathcal{I}|^c)$ for some
 constant~$c$. We call such an algorithm an FPT($x$)-algorithm.
 Under standard computational complexity assumptions, 
 fixed-parameter tractability for some parameter~$x$ is excluded when the
 problem is $\w{}\textrm{[}t\textrm{]}$-hard, $t \in \naturals$ with respect to~$x$
 or when the problem is NP-hard for a fixed value of~$x$.

 \paragraph{Committee Series Quality.}
 Let us fix a \emph{committee scoring function}~$\score \colon 2^C \rightarrow
 \mathbb{N}$, which assigns a nonnegative natural~\emph{committee score} to each
 \emph{committee}~$W \subseteq C$. Given an (ordered) \emph{committee
 series}~$\series = (W_1, W_2, \ldots, W_\timelimit)$ of~$\timelimit$ same-sized
 committees, $\util(\series) \coloneqq \sum_{i\in[\timelimit]} \score(W_i)$ is the
 \emph{utilitarian committee series quality} of~$\series$ and $\egal(\series) \coloneqq
 \min_{i\in[\timelimit]} \score(W_i)$ is its \emph{egalitarian committee series
 quality}. We study \emph{consecutive $f$-frequency committee
 series}. In such series, each candidate participates in at most~$f$
 consecutive committees of a series.

 \paragraph{Ordinal committee scoring functions.}
 Let~$E = (\candidates, \voters)$ be an arbitrary ordinal election with
 $\candidatesnr$~candidates, $\votersnr$~voters, and
 $\namedorderedsetof{W}{w}{k} \subseteq \candidates$ be a committee.
 Following~\citet{bre-kac-nie:c:successive-committees} we consider three
 (families of) ordinal committee scoring functions
 \begin{enumerate*}[label={\small(\arabic*.)}, itemjoin={{, }}, itemjoin*={{,
	 and }}]
	 \item \emph{Chamberlin-Courant} ($\cc$)
	 \item \emph{egalitarian Chamberlin-Courant} ($\egalCCScoreO$)
	 \item \emph{weakly separable scoring functions (which we denote collectively by~$\WSFamily$)}
 \end{enumerate*}
 formally defined as:
 \begin{enumerate}[label={\small(\arabic*)}, left=9pt]
	 \item $\cc(W) \coloneqq
     \sum_{i\in [\votersnr]} \max_{w \in W} \left(m - \pos_i(w) \right)$
	 \item $\egalCCScoreO(W) \coloneqq
		 \min_{i\in [\votersnr]} \max_{w \in W} \left(m -
		 \pos_i(w) \right)$
	 \item a function~$Q$ is weakly separable, i.e.,\ $Q \in \WSFamily$, if it
     can be associated with some function~$\phi \colon [\candidatesnr] \rightarrow
     \naturals_0$ such that~$Q(W) = \sum_{i\in [\votersnr]}\sum_{w\in W}
		 \phi(\pos_i(w))$
 \end{enumerate}
 The family of weakly-separable functions is very general. Among others, it
 includes such prominent voting rules as Plurality or Borda. Their respective
 $\phi$ functions are~$\phi_\textrm{plu}(x) \coloneqq \max(0, 2 - x)  $
 and~$\phi_\textrm{Bor}(x) \coloneqq  m - x$ (both defined for~$x \in
 [\candidatesnr]$).

 \paragraph{Approval committee scoring functions.}
 Similarly, for an some approval election $E = (\candidates, \voters)$ with
 $\votersnr$~voters and some committee~$W \subseteq \candidates$,
 we consider
 \begin{enumerate*}[label={\small(\arabic*.)}, itemjoin={{, }}, itemjoin*={{,
	 and }}]
	 \item \emph{(approval) Chamberlin-Courant} ($\coverageScore$)
	 \item \emph{threshold-$\alpha$ Chamberlin-Courant} ($\egalCCScoreA^\gamma$), for rational~$\gamma \in (0,1]$
	 \item \emph{approval score} ($\approvalScore$)
 \end{enumerate*}, defined formally as follows:
 \begin{enumerate}[label={\small(\arabic*)}, left=9pt]
	 \item $\coverageScore(W) \!\coloneqq\! |\{v_i \in \voters : \approvals(i) \cap W
		 \neq \emptyset \}|$
	 \item 
     $
      \egalCCScoreA^\gamma(W) \!\coloneqq\!
     \left\{
     \begin{array}{ll}
      \!\!\!1 & \!\!\!\! \mbox{if } |\{v\in \voters \!\mid\!\approvals(v) \cap W \neq
      \emptyset\}| \geq \gamma \votersnr, \\
        \!\!\!0 & \!\!\!\! \mbox{otherwise}.
     \end{array}
    \right.  $%

	 \item $\approvalScore(W) \!\coloneqq\! \sum_{v \in \voters} \approvals(v) \cap W$
 \end{enumerate}
 \paragraph{Central problem.} We focus on the following, very general,
 computational problem. We use~$\alpha \in \{\utilOp, \egalOp\}$ as a
 placeholder for a committee series quality measure and $\beta \in \{\cc,
 \egalCCScoreO, \coverageScore, \egalCCScoreA^\gamma, \approvalScore\} \cup
 \WSFamily,$ to indicate a committee scoring function. For readability, we
 directly substitute $\beta$ with $\WSFamily$, when we mean that $\beta$ is an
 arbitrary weakly-separable committee scoring function (e.g.,\
 \pname{\egalOp}{\WSFamily} is a placeholder for 
 any problem
 \pname{\egalOp}{g} where~$g \in \WSFamily$).
 We do not explicitly
 specify whether the input consists of ordinal or approval votes. It is to be
 inferred
 from
 the committee scoring function~$\beta$ in question. 

 \probDef{\pgeneralnamefull{} (\pgeneralname)}
	 {Election $E = (\candidates, \voters)$ with candidates~\candidates{} and
		voters~\voters, a number~\timelimit{} of committees in a target series, a
		size~\commsize{} of committees in a target series, a maximum candidate
		frequency~\candquota{}, and a minimal committee quality~\vallowerbound{}.
 }
 {
  Is there a consecutive \candquota{}-frequency committee
  series~\committeesSeries{} of size~$\timelimit$ consisting of size-$\commsize$
  committees such that $\alpha(\beta(\committeesSeries)) \geq \vallowerbound$?
 }
\citet{bre-kac-nie:c:successive-committees} show that \pname{\egalOp}{\beta} is
\nphard{} for all studied $\beta$ even when simultaneously $\commsize = 3$ and
$\candquota=1$. Except for~$\beta' \in \{\approvalScore, \WSFamily\}$, they
prove the same for~\pname{\utilOp}{\beta'}. Importantly, these two results
immediately exclude efficient parameterized algorithms for small values
of~$\commsize$, $\candquota$, and their sum.

Even though the above formulation is a decision problem, all our algorithms can
be used to find a requested committee series.  In theorem statements, we give
asymptotic running times of algorithms ignoring mostly irrelevant terms
polynomial in the input. For clarity, we stress it using $\OO^\star$ instead of
the standard~$\OO$. The proofs marked by~\appsymb{}, or their parts, are
deferred to the appendix.

\section{The Case of Few Candidates}
\begin{table*}[t]%
  \setlength{\tabcolsep}{2pt}
  \begin{tabular}{rllp{-0pt}rll}
       \cmidrule[\heavyrulewidth]{1-3}
       \cmidrule[\heavyrulewidth]{5-7}
        & candidates number~$m$ & & & & time horizon~$\tau$ (for const.~$\commsize{}$) &\\
       \cmidrule{1-3}
       \cmidrule{5-7}
     	  \pname{\egalOp}{\beta}
     		& $\OO^\star(2^\candidatesnr)$ 
        & \footnotesize{\Shcref{thm:poly-mult-f1}} & 
        & \pname{\egalOp}{\beta}
        & $\OO^\star(2.851^{(k-0.5501)\timelimit})$
        & \footnotesize{\Shcref{thm:egal-fpt-tf1}}
     	\\
     	  \pname{\utilOp}{\beta}
     	  & $\OO^\star(\candidatesnr!(\commsize-1)^\candidatesnr)$
        & \footnotesize{\Shcref{thm:fpt-m}} & 
        & \pname{\utilOp}{\beta}
    	  & $\OO^\star(2.851^{(k-0.5501)\timelimit})$; $\candquota = 1$
    	  & \footnotesize{\Shcref{thm:util-fpt-tf1}}
     	\\
       	& $\OO^\star(2^\candidatesnr)$; const.~$\commsize$ and $\candquota = 1$
        & \footnotesize{\Shcref{thm:poly-mult-f1}} &
    		& 
        & $\OO^\star(2^{k\timelimit(f+1)}(2e)^{k\timelimit}(k\timelimit)^{\log
        (k\timelimit)})$; const.~$\candquota$
    		& \footnotesize{\Shcref{thm:util-fpt-tf2}}
     	\\
     	& $\OO^\star(4^{\candquota\candidatesnr})$; any~$\candquota$
     	& \footnotesize{\Shcref{thm:util-fpt-mf}}
     	\\
     \cmidrule[\heavyrulewidth]{1-3}
     \cmidrule[\heavyrulewidth]{5-7}
     \end{tabular}%
	\caption{Summary of our results for the egalitarian and utilitarian committee
		series quality functions for all $\beta \in \{\cc, \egalCCScoreO,
		\coverageScore, \egalCCScoreA^\gamma, \approvalScore\}$. In all cases we
		assume $k \geq 3$. Note that~\citet{bre-kac-nie:c:successive-committees}
		show polynomial-time algorithms for \pname{\utilOp}{\approvalScore} and
	  \pname{\egalOp}{\WSFamily}, so in these cases our algorithms run (much)
    slower.\label{tab:results}}
\end{table*}

Given the general computational hardness results
of~\citet{bre-kac-nie:c:successive-committees}, we start our search of efficient
algorithms from cases with a bounded number of candidates. As argued in
the introduction, this assumption can naturally be justified from the practical
point of view. Parameter ``number of candidates,'' which we denote
by~$\candidatesnr$, is too a standard parameter in the literature on the
complexity of
election problems.

For better accuracy, in some subsequent results (and in~\Cref{tab:results}) we
give running times using the size~$\commsize$ of committees. Because~$\commsize
\leq \candidatesnr$, such results always yield fixed-parameter tractability for
parameter~$\candidatesnr$. Naturally, they also show fixed-parameter
tractability for parameter~$\candidatesnr + \commsize$. Recall that the hardness
results of~\citet{bre-kac-nie:c:successive-committees} exclude fixed-parameter
tractability for~$\commsize$ alone.

Finding the winner of every multiwinner voting rule that assigns a score to each
committee and chooses the one with the maximum score as the winner is
fixed-parameter tractable with respect to the number of candidates (assuming
computing the score of a committee is polynomial-time solvable). This
observation becomes easy, when one realizes that for such rules it is enough to
enumerate all possible committees and compute their scores
(for example, see the works of~\citet{pro-ros-zoh:c:proporional-repr-complexity,bet-sli-uhl:j:mon-cc-multivariate}) . However, a committee
series is composed of multiple committees whose interdependency is nontrivially
governed by the frequency of candidates and the requirement of consecutiveness.
Indeed, even if we enumerate all at most~$2^\candidatesnr$ committees, to
construct a committee series we need to consider that each of them possibly come
in any between one and frequency-many copies as demonstrated in Example~3 by
\citet{bre-kac-nie:c:successive-committees}.

Does this complication to the domain of all possible solutions make
finding the right solution computationally harder with respect to parameter
``number of candidates''?

For the good news, we answer the above in negative. Quite surprisingly,  even
though the number of all possible committee series depends nonlinearly on the
candidate frequency, the following series of results show fixed-parameter
tractability of the problems we study with respect to solely the number of
candidates. By this, the results reveal that the increase of the complexity of the
domain of solutions is, intuitively, still bounded by a function of the number
of candidates.

We start with a foundational result for the case of~$\candquota = 1$. While the
result applies to both the utilitarian and egalitarian variants of our problem,
it is the latter variant for which the result gives fixed-parameter
tractability.

\begin{restatable}[\appsymb]{theorem}{polymultfone}
  There  exists an algorithm that solves \pname{{\utilOp}}{\beta} and
  \pname{{\egalOp}}{\beta} in $\OO^\star(2^m)$ time for $f=1$ and all
  studied~$\beta$.
\label{thm:poly-mult-f1}\end{restatable}

The importance of~\Cref{thm:poly-mult-f1} for the \pname{\egal}{\beta} problem
lies in combining it with the subsequent observation
by~\citet{bre-kac-nie:c:successive-committees}.

\begin{restatable}[Lemma 1, \citealt{bre-kac-nie:c:successive-committees}]{proposition}{prop:egal2-egal1}
  \pname{\egal}{\beta} with $f\geq 2$  can be reduced to \pname{\egal}{\beta}
  with $f=1$ in linear time.  
\label{prop:egal2-egal1}\end{restatable}

The reduction in~\Cref{prop:egal2-egal1} does not increase the number of
candidates. Hence, with~\Cref{thm:poly-mult-f1}, \Cref{prop:egal2-egal1} yields
a general result about the egalitarian version of the problem, which shows the
sought tractability for small number of candidates.

\begin{restatable}{theorem}{thm:egal-poly-mult}
  There exists an algorithm that solves \pname{\egal}{\beta}  in
  $\OO^\star(2^m)$ time for all studied~$\beta$.
\end{restatable}

Generalizing~\Cref{thm:poly-mult-f1} in the utilitarian case does not directly
lead to such an optimistic result as that for the egalitarian case.
Instead, we obtain dynamic programming (DP) algorithms whose running time
upper-bounds increase exponentially with the increase of
parameter~$\candquota$.
To increase readability, we first state our result for $\candquota = 2$.  

\begin{theorem}\label{thm:pol-util-ftwo}
  There is an $\fpt(\candidatesnr)$-algorithm running in time
  $\OO^\star(4^{2\candidatesnr})$ solving
  \pname{\util}{\beta} for $\candquota = 2$ and all studied~$\beta$.
\end{theorem}
\begin{proof}
	Let us define a Boolean function~$F(A, j, s, G)$, which returns
  \true{} if among candidates~$A$, there is a series
    achieving quality at least~$s$ and consisting of~$j$~committees,
    each of size $2$ such that
    in the~$j$-th committee 
	candidates from~$G \subseteq A$ appear for the first time
  in the series; $F$ returns \false{}
	otherwise. Our algorithm computes the values of~$F$ and returns
  \true{}, if at
	least one of~$F(C, \tau, \vallowerbound, G')$, for~$G' \subseteq C$, $|G'| \leq k$ is
  \true{}.

  We compute the values of~$F$ applying the DP approach to the following
  recursive formula. Denoting by~$\mathcal{W}(A, G, k)$ a collection of all
  size-$\commsize$ sets~$W$ such that $G \subseteq W \subseteq A$ we get:
	\begin{align*}
	   f(A, &j, s, G) = \\
      &\bigvee_{W \in \mathcal{W}(A, G, k)} f(A \setminus G, j-1,
	s-\score(W), W \setminus G). 
	\end{align*}
	The correctness of the formula follows from the
    equation. Recall that to compute the left hand side value, we need to
    check whether it is possible to obtain exactly~$j$
	$2$-consecutive committees yielding quality~$s$ such that candidates from~$G$
	appear in the $j$-th committee for the first time. To do so, in the formula,
    we consider each
	possible committee~$W$ containing~$G$ as a subset and scoring~$\score(W)$. For
	each such $W$, we check whether there is a series of $j-1$ 
    consecutive~committees of the requested size that jointly get
    quality~$s-\score(W)$ consisting of candidates in~$A \setminus G$. The last
    conditions follows from the fact that candidates in~$G$ must be used for
    the first time only in $W$ in question.
	Furthermore, we assure that the $j-1$-th committee contains the candidates
	that appear in~$W$ for the second time, as these have to take part in the
	committees consecutively; hence the last argument~$W \setminus G$ of~$F$'s
	evaluation on the right-hand side. To conclude the proof,
  we define $F(A, 0, s, G) = \textrm{\true}$
  for $s \leq 0$ and~\false{} for~$s > 0$. We let the value of the
  function be~\false{} each time $\mathcal{W}(A, G, k) = \emptyset$.

    There are at most $\timelimit2^\candidatesnr\textrm{poly}(\votersnr,
    \candidatesnr)$ values of~$F$ to compute, where the polynomial term comes
    form the maximum score of our committee scoring functions. It takes at
    most~$2^\candidatesnr$ to compute a single value. Hence, as a result we
    obtain the claimed running time of~$\OO^\star(4^{2\candidatesnr})$.
\end{proof}
The above approach readily generalizes to arbitrary values of~$\candquota$ by
defining $F$ to take arguments~$G_1$ up to~$G_{\candquota-1}$ instead of
just~$G$.
Increasing the parameter space, leads to the exponential in~$\candquota$
increase of the running time.
\begin{restatable}[\appsymb]{theorem}{utilfptmf}\label{thm:util-fpt-mf}
  There is an $\fpt(\candidatesnr)$-algorithm running in
  time~$\OO^\star(4^{\candquota\candidatesnr})$ that solves
	 \pname{\util}{\beta} for all studied~$\beta$.
\end{restatable}

In the remainder of this section, we present~\Cref{thm:fpt-m} covering the
claimed fixed-parameter tractability for $\candidatesnr$ for
\pname{\util}{\beta}. There is a clear theoretical advantage of this next result
over the one from~\Cref{thm:util-fpt-mf}, as the latter depends exponentially on
both~$\candidatesnr$ and~$\candquota$. As we shall see, however, the running
time of the algorithm from~\Cref{thm:fpt-m} may be (asymptotically) as large as
$\candidatesnr^\candidatesnr$. Such a running time would (asymptotically) be way
larger than $4^{\candquota\candidatesnr}$ coming from~\Cref{thm:util-fpt-mf} for
small values of~$\candquota$. Hence, the algorithm from~\Cref{thm:util-fpt-mf}
(as well as the one from~\Cref{thm:poly-mult-f1}) might be more useful in
practice. Especially since one of the motivations of studying candidate
frequencies is to avoid excessive overload of committee members, which leads to
small values of~$\candquota$ \citep{bre-kac-nie:c:successive-committees}.
Finally, the algorithm from the following~\Cref{thm:fpt-m}, albeit interesting
from the computational complexity classification perspective, would rather turn
out too computationally intensive to be applicable in practice.

\begin{theorem}\label{thm:fpt-m}
	There is an $\fpt(\candidatesnr)$-algorithm which runs in time
	$\OO^\star(\candidatesnr!(\commsize+1)^\candidatesnr)$ and solves
	\pname{\alpha}{\beta} for $\alpha \in \{\utilOp, \egalOp\}$ and all
  studied~$\beta$.
\end{theorem}
\begin{proof}
  Throughout the whole proof we let $\commsize$ be the size of committees in a
  requested $\candquota$-consecutive committee series of size~$\timelimit$. Our
  algorithm computes the sought committee series by repeatedly running a dynamic
  programming (DP) procedure over a collection of guesses, each of which represents a
  subspace of the solution space. The proof comes in three parts. We first
  discuss the guesses, then the DP procedure, and finally we show how to
  effectively enumerate the guesses.

  \paragraph{Guesses.}
  Our algorithm repeatedly guesses a~\emph{division}, which is a specific
  structure over candidates~\candidates, which describes a space of possible
  $\candquota$-frequency committee series of size~$\timelimit$.
  Before we describe divisions, let us first consider some order of candidates
  represented by a bijection~$\rho \colon \left[\candidatesnr\right] \rightarrow
  \candidates$; e.g.,\ $\rho(2) = c'$ means that candidate~$c' \in \candidates$ is
  ordered second. A subset~$X$ of candidates \emph{is an interval} (according
  to~$\rho$) if the set $\{i \colon \rho(i) \in X\}$ is an interval;
  intuitively, the candidates from~$X$ form an interval according to~$\rho$.
  If~$X$ is an interval, then we denote by~$\beg(X)$ the position in~$\rho$ of
  the leftmost candidate of~$X$.
  \begin{definition}
    A division~$\PSdivision$ of size~$\PSdivisionSize$ is a $3$-tuple $(\rho,
    \PSdivsets, \PSdivfil)$ with \emph{order} $\rho$ of candidates, an ordered
    set~$\PSdivsets = \{S_1, S_2, \ldots S_\PSdivisionSize\}$ of
    size-$\commsize$ committees called \emph{primitives}, a
    collection~$\PSdivfil = \{F_1, F_2, \ldots, F_\PSdivisionSize\}$ of
    \emph{interim candidate sets} such that jointly:
    \begin{enumerate}[nosep, label=\alph*), left=9pt]
     \item $S_i$ is an interval according to~$\rho$ for each $i \in [d]$,
     \item $\beg(S_i) < \beg(S_{i+1})$ for each $i \in [d-1]$,
     \item $F_\PSdivisionSize =
      \emptyset$,
     \item $F_i \cap F_j =
       \emptyset$ for each $(i, j) \in [d] \times [d]$, $i \neq j$,
     \item $|F_i| \in \{0,\commsize - |S_i
       \cap S_{i+1}|\}$ for each $i \in [\PSdivisionSize - 1]$,
     \item $S_i \cap S_{i+1} = \emptyset \Rightarrow F_i = \emptyset$ for each
       $i \in [\PSdivisionSize - 1]$, and
    \item $\left(\bigcup_{i \in [\PSdivisionSize]}F_i\right) \cap \left(\bigcup_{i \in
      [\PSdivisionSize]}S_i\right) = \emptyset$.
    \end{enumerate}
  \end{definition}
  Together with the following intuitive description of the above definition, we
  provide an example division in~\Cref{ex:division}. Let us fix some
  order~$\rho$. Then, the definition says that the primitives in~$\PSdivsets$
  are $\PSdivisionSize$ pairwise different interval size-$\commsize$ committees
  ordered increasingly by their interval beginning according to~$\rho$. Further,
  the interim candidate sets are pairwise disjoint and consist of candidates
  that do not belong to any primitive committee. Interim candidate
  set~$F_\PSdivisionSize$ is always empty, and each interim candidate set~$F_i$,
  $i \in [d-1]$ can be non-empty only if $Y = S_i \cap S_{i+1}$ is non-empty;
  but if~$F_i$ is non-empty, then it contains~$\commsize - |Y|$ elements.  

  A division is intended to encode a (sub)space of feasible solutions
  that can be used for an efficient search for
  a solution to our problem. So, it is crucial that at least one division
  describes a space containing the solution, if the latter exists.
  \begin{restatable}[\appsymb]{lemma}{divisionspace}\label{lem:division-space}
    For each instance of~$\pname{\alpha}{\beta}$, there is a
    division~$\PSdivision$ whose space contains a committee series maximizing
    $\alpha(\beta(\series))$ for the given values of~$\candquota$, $\timelimit$,
    and~$\commsize$. 
  \end{restatable}
  The proof of the above lemma is constructive, but it does not offer an
  efficient way of finding an optimal solution. Before showing such
  a way, we present an intuition of how the elements of the space of some
  division~$\PSdivision$ look like. To ease the presentation, we assume the
  identity order~$\rho(i) = c_i$ for $i \in \left[\candidatesnr\right]$, and lay
  out some elements of the space in~\Cref{ex:comm-series}.

  First, observe that some consecutive $\candquota$-frequency committee of
  size~$\timelimit$ can be constructed as follows. We assign each
  primitive~$S_i \in \PSdivsets$ the respective number~$r_i$ of copies of
  this primitive ensuring that the sum of~$r_i$'s is exactly~$\timelimit$.
  Assuming that our~$r_i$'s do not violate the frequency limit~$\candquota$, we
  get a committee series~$\series$ of size~$\timelimit$ by repeating, in order,
  each primitive~$S_i$ exactly~$r_i$ times. 
  Clearly, $\series$ is consecutive due to the interval requirements on the
  primitives in~$\PSdivsets$.
  There are, however, more committee series similar to~$\series$ in the space
  of~$\PSdivision$. They use the so-far ignored interim candidate
  sets~$\PSdivfil$.
  For an example, consider the following case. Assume that~$\series$ contains a
  sequence $S_i, S_i, S_{i+1}$ of neighboring committees and
  that~$Y=S_i \cap S_{i+1}$ is non-empty.
  If $F_i \neq \emptyset$, then we can substitute the ``middle'' $S_i$,
  with a committee $S' = Y \cup F_i$, thus
  obtaining a new committee series~$\series'$ containing a
  sequence~$S_i, S', S_{i+1}$. Note that $\series'$ remains
  $\candquota$-consecutive
  because by definition $|F_i| = \commsize - |Y|$ and~$F_i \cap Y = \emptyset$.

  The whole space of~$\PSdivision$ consists of the committee series emerging
  from all \emph{valid} choices of the multiplicities~$r_i$ and from all valid
  substitutions (perhaps multiple at once) analogous to the construction
  of~$\series'$.

	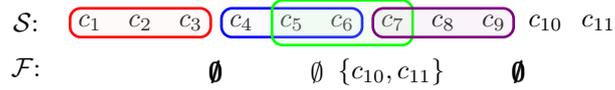
\begin{figure}\centering
		\begin{tikzpicture}[
			cand/.style={draw=none,fill=none, minimum width = 1.85em},
			highligh/.style={draw=none, rounded corners, fill opacity=0.20, inner xsep=-2px,
			inner ysep=-1px, line width = 1px}]
        \def\n{11}
    
			  \node[cand] (cand-1) at (0, 0) {$c_{1}$};
        \foreach \i in {2, ..., \n} {
					\pgfmathparse{int(\i-1)}
					\node[cand, right = 0px of cand-\pgfmathresult] (cand-\i) {$c_{\i}$};
        }

				\node[highligh, fill=red!10!white, draw=red, fit = (cand-1) (cand-3)] {}; 
				\node[highligh, fill=blue!10!white, draw = blue, fit = (cand-4) (cand-6)] {}; 
				\node[highligh, fill=green!20!white, draw = green, fit = (cand-5)
          (cand-7), inner ysep = 2px] {}; 
				\node[highligh, fill=violet!20!white, draw = violet, fit = (cand-7)
          (cand-9)] {}; 

        \node[cand, left = 5px of cand-1] (divs) {$\PSdivsets$:};

			  \node[cand, below right = 2ex and 0em of cand-1, anchor = north] (fcand-1) {};
        \foreach \i in {2, ..., \n} {
					\pgfmathparse{int(\i-1)}
					\node[cand, right = 0px of fcand-\pgfmathresult] (fcand-\i) {};
        }

        \node[cand] (f_1) at (fcand-3) {$\pmb{\emptyset}$};
        \node[cand] (f_2) at (fcand-5) {$\emptyset$};
        \node[cand, xshift = -1em] (f_3) at (fcand-7) {$\{c_{10}, c_{11}\}$};
        \node[cand] (f_4) at (fcand-9) {$\pmb{\emptyset}$};

        \node[cand, below = 3px of divs] (divf) {$\PSdivfil$:};

    \end{tikzpicture}
    \caption{Division~$\PSdivision = (\PSorder, \PSdivsets, \PSdivfil)$ of
      size~$d=4$ of
      $11$~candidates for a committee size~$\commsize=3$ and
      order~$\PSorder(i) = c_i$. The division features $\PSdivsets = (\{c_1, c_2,
      c_3\},\{c_4, c_5, c_6\}\{c_5, c_6, c_7\}, \{c_7, c_8, c_9\})$ and
      $\PSdivfil = \{\pmb{\emptyset}, \emptyset, \{c_{10}, c_{11}\},
      \pmb{\emptyset}\}$. Bold~$\pmb{\emptyset}$ are empty by definition.
    }\label{ex:division}
	\end{figure}

  \paragraph{Finding a Solution.}
  Assume a division~$\PSdivision = (\PSorder, \PSdivsets,
  \PSdivfil)$, where $\PSdivsets = (S_1, S_2, \ldots, S_{\PSdivisionSize})$
  and~$\PSdivfil = (F_1, F_2, \ldots, F_\PSdivisionSize)$.
  We show a dynamic programming algorithm (DP) running in the claimed FPT-time that
  tests whether the space of~$\PSdivision$ contains a sought committee.
  We define a Boolean function~$T(t,i,q,u,j,r,p)$, for $q>0$, $r>0$, which
  is \true{} if and only if there is an $\candquota$-consecutive series $X$
  such that:
	\begin{enumerate}[nosep, label=\alph*), left=9pt]
		\item $X$ contains $t$~committees (of size~$\commsize$),
    \item primitive~$S_i \in \PSdivsets$ is repeated exactly $q$~times
      in~$X$,
    \item if~$F_i \neq \emptyset$, committee~$S' = F_i \cup (S_i \cap S_{i+1})$,
      called \emph{$i$-intermediate}, is repeated $u$~times in~$X$ between the
      copies of~$S_i$ and~$S_{i+1}$,
		\item in which candidate~$\PSorder(j) \in S_i$ is part of exactly
			$r$~consecutive committees of~$X$, and
		\item $\alpha(\beta(X)) \geq p$.
	\end{enumerate}
  Assuming that $\PSorder(j')$ is the last candidate in~$S_\PSdivisionSize$ according
  to~$\PSorder$, there exists a sough $\candquota$-consecutive series of
  size~$\timelimit$ achieving a committee series quality~$\eta$ if there is a
  pair~$(q', r') \in [\candquota] \times [\candquota]$ for which~$T(\timelimit,
  \PSdivisionSize, q', 0, j', r', \eta)$ outputs \true{}.

  We recursively compute the values~$T(t,i,q,u,j,r,p)$ in the order of
  increasing values of~$t$, $i$, $j$, and~$p$; the variables are mentioned in
  the order of increasing frequency of change during computation iterations.
  Because the candidates are ordered with respect to~$\PSorder$, it is convenient to
  think of combinations of values~$i$ and~$j$ as cells in a grid, where $j$ is
  the column number representing candidate~$\PSorder(j)$ and $i$ is the row number
  representing primitive~$S_i \in \PSdivsets$. In the grid, primitive~$S_i$
  forms an interval in the $i$-th row; this interval corresponds to the
  candidates in~$S_i$. As depicted in~\Cref{fig:grid-representation}, by
  convention, we place point~$(1,1)$, corresponding to candidate~$\PSorder(1)$ in
  primitive~$S_1$, in the upper left corner of the corresponding grid.

  Some values of $T(t,i,q,u,j,r,\eta)$ are invalid by definition or in an
  obvious manner. This is the case for:
	\begin{enumerate}[nosep, label=(I\arabic*), left=9pt]
    \item each $i$ and $j \in [\candidatesnr]$ for which~$\PSorder(j) \notin
      C_i$; see dark shaded cells in~\Cref{fig:grid-representation},
    \item cases in which $q+t > \tau$, $q + u > \candquota$, or $r > \candquota$,
    \item positive values~$u$ for $i$ for which $F_i=\emptyset$; in particular,
       by definition of~$\PSdivision$, this holds when $S_i \cap S_{i+1} \neq
      \emptyset$\label{enum:invalid:posu}.
  \end{enumerate}
  For readability, we never explicitly test for these invalid cases assuming,
  for technical reasons, that $T$ returns \false{} for them.

  For the sake of presentation, for each primitive~$S_i \in \PSdivsets$, we
  define the \emph{head} and the \emph{tail} of~$S_i$ as, respectively, the
  first and the last candidate of~$S_i$ with respect to order~$\rho$
  (as demonstrated in~\Cref{fig:grid-representation}). Formally, $\head(S_i) =
  j'$ and $\tail(S_i) = j''$ if for each $j \in [\candidatesnr]$ such that
  $\rho(j) \in S_i$, $j'\leq j$ and $j\leq j''$. If a candidate is
  neither the tail nor the head of some primitive~$S_i \in \PSdivsets$, then it
  is \emph{intermediate} in this primitive. Note that a candidate might be the
  head of one primitive but an intermediate candidate for another primitive. 

  In what follows, we frequently consider a situation in which we seek a
  committee series of quality at least~$p$ by repeating a primitive~$S_i \in
  \PSdivsets$ exactly $q$~times and the committee $(S_i \cap S_{i+1}) \cup F_i$
  exactly $u$~times.
  In such cases, we use~$\reqqual(i,p,q,u) \coloneqq p - q\beta(S_i) -
  u\beta((S_i \cap S_{i+1})\cup F_i)$ that describes the required quality of a
  to-be-constructed committee series prior to using the $q + u$ committees
  mentioned above. This is where the subtraction in the expression comes
  from. Our formulas use~$u=0$ whenever $(S_i \cap S_{i+1})\cup F_i$ is not a
  committee of the requested size (when either~$F_i = \emptyset$ or $S_i \cap
  S_{i+1} = \emptyset$). Consequently, the related term vanishes.

	\begin{figure}\centering
		\begin{tikzpicture}[
			cand/.style={draw=none,fill=none, minimum width = 1.5em},
			highligh/.style={draw=none, rounded corners, fill opacity=0.2, inner xsep=-2px,
			inner ysep=-1px}]
    
			  \node[cand] (c1-cand-1) at (0, 0) {$c_{1}$};
			  \node[cand, below = 0px of c1-cand-1] (c1-cand-2) {$c_{2}$};
			  \node[cand, below = 0px of c1-cand-2] (c1-cand-3) {$c_{3}$};

			  \node[cand, right = 0px of c1-cand-1] (c2-cand-1) {$c_{1}$};
			  \node[cand, below = 0px of c2-cand-1] (c2-cand-2) {$c_{2}$};
			  \node[cand, below = 0px of c2-cand-2] (c2-cand-3) {$c_{3}$};

			  \node[cand, right = 0px of c2-cand-1] (c3-cand-1) {$c_{4}$};
			  \node[cand, below = 0px of c3-cand-1] (c3-cand-2) {$c_{5}$};
			  \node[cand, below = 0px of c3-cand-2] (c3-cand-3) {$c_{6}$};
			  \node[cand, right = 0px of c3-cand-1] (c4-cand-1) {$c_{7}$};
			  \node[cand, below = 0px of c4-cand-1] (c4-cand-2) {$c_{5}$};
			  \node[cand, below = 0px of c4-cand-2] (c4-cand-3) {$c_{6}$};

			  \node[cand, right = 0px of c4-cand-1] (c5-cand-1) {$c_{7}$};
			  \node[cand, below = 0px of c5-cand-1] (c5-cand-2) {$c_{10}$};
			  \node[cand, below = 0px of c5-cand-2] (c5-cand-3) {$c_{11}$};
        
			  \node[cand, right = 0px of c5-cand-1] (c6-cand-1) {$c_{7}$};
			  \node[cand, below = 0px of c6-cand-1] (c6-cand-2) {$c_{8}$};
			  \node[cand, below = 0px of c6-cand-2] (c6-cand-3) {$c_{9}$};

			  \node[cand, right = 4em of c6-cand-1] (d1-cand-1) {$c_{1}$};
			  \node[cand, below = 0px of d1-cand-1] (d1-cand-2) {$c_{2}$};
			  \node[cand, below = 0px of d1-cand-2] (d1-cand-3) {$c_{3}$};

			  \node[cand, right = 0px of d1-cand-1] (d2-cand-1) {$c_{4}$};
			  \node[cand, below = 0px of d2-cand-1] (d2-cand-2) {$c_{5}$};
			  \node[cand, below = 0px of d2-cand-2] (d2-cand-3) {$c_{6}$};

			  \node[cand, right = 0px of d2-cand-1] (d3-cand-1) {$c_{4}$};
			  \node[cand, below = 0px of d3-cand-1] (d3-cand-2) {$c_{5}$};
			  \node[cand, below = 0px of d3-cand-2] (d3-cand-3) {$c_{6}$};
			  \node[cand, right = 0px of d3-cand-1] (d4-cand-1) {$c_{7}$};
			  \node[cand, below = 0px of d4-cand-1] (d4-cand-2) {$c_{5}$};
			  \node[cand, below = 0px of d4-cand-2] (d4-cand-3) {$c_{6}$};

			  \node[cand, right = 0px of d4-cand-1] (d5-cand-1) {$c_{7}$};
			  \node[cand, below = 0px of d5-cand-1] (d5-cand-2) {$c_{5}$};
			  \node[cand, below = 0px of d5-cand-2] (d5-cand-3) {$c_{6}$};
        
			  \node[cand, right = 0px of d5-cand-1] (d6-cand-1) {$c_{7}$};
			  \node[cand, below = 0px of d6-cand-1] (d6-cand-2) {$c_{8}$};
			  \node[cand, below = 0px of d6-cand-2] (d6-cand-3) {$c_{9}$};

				\node[highligh, fill=red!20!white, draw = red, fit = (c1-cand-1) (c1-cand-3)] {}; 
				\node[highligh, fill=red!20!white, draw = red, fit = (c2-cand-1) (c2-cand-3)] {}; 
				\node[highligh, fill=blue!20!white, draw = blue, fit = (c3-cand-1) (c3-cand-3)] {}; 
				\node[highligh, fill=green!40!white, draw = green, fit = (c4-cand-1) (c4-cand-3)] {}; 
				\node[highligh, fill=green!40!white, draw = green, inner xsep = 0, inner ysep = 0, fit = (c5-cand-1)] {}; 
				\node[highligh, fill=violet!40!white, draw = violet, fit = (c5-cand-1)] {}; 
				\node[highligh, fill=violet!40!white, draw = violet, fit = (c6-cand-1) (c6-cand-3)] {}; 

				\node[highligh, fill=red!20!white, draw = red, fit = (d1-cand-1) (d1-cand-3)] {}; 
				\node[highligh, fill=blue!20!white, draw = blue, fit = (d2-cand-1) (d2-cand-3)] {}; 
				\node[highligh, fill=blue!20!white, draw = blue, fit = (d3-cand-1) (d3-cand-3)] {}; 
				\node[highligh, fill=green!40!white, draw = green, fit = (d4-cand-1) (d4-cand-3)] {}; 
				\node[highligh, fill=green!40!white, draw = green, fit = (d5-cand-1) (d5-cand-3)] {}; 
				\node[highligh, fill=violet!40!white, draw = violet, fit = (d6-cand-1) (d6-cand-3)] {}; 
    \end{tikzpicture}
    \caption{Two selected $3$-consecutive $3$-committee series of
      size~$\timelimit=6$ from the space of division~$\PSdivision$
      of~\Cref{ex:division}. Each column represents a single committee in a
      committee series, and each color represents a different primitive (or its part).}\label{ex:comm-series}
	\end{figure}
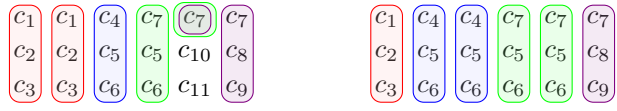

  We continue with presenting a recursive formula for~$T$, which we split into
  multiple cases for readability.
  Values of the base cases~$T(t, 1, q, u, j, r, p)$ follow directly from the
  definition of~$T$.

  We first focus on a case~$T(t,i,q,u,j,r,p)$ when $\head(S_i) = j$ and when
  $S_{i-1} \cap S_i= \emptyset$. We skip computation at all and return \false{}
  if~$q \neq r$. Indeed, $q=r$ is required as all candidates in $S_i$ will
  appear for the first time in the constructed committee series.
  \small%
	\begin{align*}
    T(t,i,q,u,j,r,p) = \bigvee_{q' \in [\candquota], r' \in
    [\candquota]} T&(t-q-u, i-1, q', 0,\\
                   & \tail(S_{i-1}), r', \reqqual(i, p, q, u))
	\end{align*}
  \normalsize%
  In words, the formula's alternative tests whether we can achieve sufficient
  committee quality~$\reqqual(i, p, q, u)$, \emph{before} we repeat
  primitive~$S_i$ exactly~$q$ times and committee $S_{i} \cap S_{i+1} \cup F_i$
  exactly~$u$~times. Since $S_{i-1} \cap S_{i} = \emptyset$ implies that
  $F_{i-1} = \emptyset$, in the alternative we invoke~$T$ with the respective
  parameter equal to~$0$. Recall that expression $\reqqual(\cdot)$ remains
  correct even if $F_i = \emptyset$, which by Case~\ref{enum:invalid:posu}
  implies $u=0$.  

  Next, we consider $\head(S_i) = j$ under assumption that $S_{i-1} \cap S_{i}
  \neq \emptyset$. Consequently, $q<r$, as the $j$-th candidate belongs to
  $S_{i-1} \cap S_{i}$, and $S_{i-1}$ has to be repeated at least once (see the
  domain of~$T$). So, if $q \geq r$ we return \false{} and otherwise we compute~$T$
  as follows:
  \small%
	\begin{align*}
    T(t,i,q,u,j,r,p) = \bigvee_{q' \in [\candquota], u' \in
    [\candquota]}
                       T(&t-q-u, i-1, q', u', j,\\
			&r - q, \reqqual(i, p, q, u)).
	\end{align*}
	\normalsize%
  This time, we ensure that before using primitive~$S_i$ and interim
  candidates~$F_i$, we used the $j$-th candidate exactly $r-q$~times.  We do not
  subtract~$u$, as by the fact that primitives are different, $j$-th candidate
  cannot be part of $(S_{i} \cap S_{i+1}) \cup F_i$, which is repeated
  $u$~times.
  Accordingly, the alternative iterates over all possible values of~$q'$
  and~$u'$ to check if there is any committee series using primitives up
  to~$S_{i-1}$ of the required quality (see the discussion in the previous
  case).

  We continue with computing $T(t,i,q,u,j,r,p)$ for cases, in which the $j$-th
  candidate is either an intermediate or a tail candidate. Consequently,
  the $j$-th candidate (from the respective primitive~$S_i \in \PSdivsets$) can
  be part of $\ell$-intermediate committees for $\ell \geq i$, i.e.,\ that come
  after all copies of~$S_i$ in the built committee series~$\series$. Previously,
  when the $j$-th candidate was a head candidate, it could only be part of the
  $\ell$-intermediate committees for~$\ell < i$. For this reason, the following
  formulas sometimes impose additional conditions. 
  
  We start with a simple case. Let the $j$-th candidate $c = \rho(j)$ be the
  intermediate or tail of primitive~$S_i$. Assume that  $c \notin (S_i \cap
  S_{i+1})$.
  \small%
	\begin{align*}
		T(t&,i,q,u,j,r,p)=
    \bigvee_{r' \in \{r, r+1, \ldots, \candquota\}} T(t, i, q, u, j-1, r', p)
	\end{align*}
  \normalsize%
  The number repetitions of $S_i$ and the score~$p$ is always fixed
  during the computation for $S_i$'s head candidate. So in this case the
  alternative recursively mostly ``carries on'' function values from previous candidates
  of~$S_i$. We iterate using $r'$ because the $j-1$-th candidate could be part
  of more (previous) committees than candidate~$c$.
  Due to our interval requirements and sorting
  of~$S_i \in \PSdivsets$, we know that $r' \geq r$.
  
  Next, consider computing $T(t,i,q,u,j,r,p)$ assuming that the $j$-th
  candidate~$c = \rho(j)$ is an intermediate or tail candidate, $c \in S_i \cap
  S_{i+1}$, and $c \notin S_{i-1}$. We repeat~$c$ exactly $q+u$~times as a
  member of $q$~copies of~$S_i$ and of $u$~copies of $i$-intermediate committee.
  Hence, we immediately return \false{} if~$r \neq q+u$, and otherwise we have
  \small%
	\begin{align*}
		T(t&,i,q,u,j,r,p)=
    \!\!\bigvee_{r' \in \{q, q+1, \ldots, \candquota\}}
      \!\!T(t, i, q, u, j-1, r', p).
	\end{align*}
  \normalsize%
  The lower-bound of iterator $r'$ follows from the fact that the $j-1$-th
  candidate must be repeated at least as many times as the primitive~$S_i$ to
  which the candidate belongs is repeated.  

  In the final case we consider the, intermediate or tail, $j$-th candidate $c =
  \rho(j)$ such that $c \in S_i \cap S_{i+1}$, and $c \in S_{i-1}$. Since~$c$
  is used in copies of primitive~$S_{i-1}$, we return \false{} if $q + u \leq r$.
  Otherwise, we compute $T$ as follows:

  \small%
	\begin{align*}
    T&(t,i,q,u,j,r,p) =
    \!\!\bigvee_{r' \in \{q, q+1, \ldots, \candquota\}}\!\!
        T(t, i, q, u, j-1, r', p) \land  \\
       & 
       \!\!\!\bigvee_{q' \in [\candquota], u'\in [\candquota]}
       \!\!\!T(t-q-u, i-1, q', u', j, r - q - u, \reqqual(i, p, q, u)).
	\end{align*}
  \normalsize%

  On top on all conditions of the counterpart case with~$c \notin S_{i-1}$
  described directly above this case, we add further requirements. Observe that
  here $c$ is used $q+u$~times but it was also previously used. The second
  alternative of the formula ensures that these $q+u$~uses of candidate~$c$
  would not make~$c$ exceed the frequency~$\candquota$.

  The case distinctions is exhaustive so the correctness of the formula stems
  directly from its description. To compute each value of function~$T$, we need
  at most $O(\candquota^2)$~steps.  Further, for some division~$\PSdivision$ of
  size~$\PSdivisionSize$, we need to compute $\PSdivisionSize \cdot
  \commsize\timelimit\candquota^3 \cdot \textrm{poly}(\votersnr,
  \candidatesnr)$~values of~$T$, as there is~$\PSdivisionSize \commsize$~valid
  pairs of parameters~$i$ and~$j$, $t \in [\timelimit]$, $(q, u, r) \in
  [\candquota]^3$. Hence, we obtain a polynomial running-time for each division.
  \begin{figure}\centering
		\begin{tikzpicture}[
			cand/.style={draw=none,fill=none},
			highligh/.style={draw=none, opacity=0.75, inner xsep=-1px,
			inner ysep=-1px},
			tight matrix/.style={every outer matrix/.append style={inner sep=+0pt}}]

		  \matrix (m) [matrix of nodes, nodes in empty cells,
		  	row sep = -\pgflinewidth,
				column sep = -\pgflinewidth,
				nodes={text width=1.2em, minimum size = 1em, anchor=center,
			  	draw = black!20!white, outer sep = 0em, inner sep = 0em, align =
		    	center},
				row 5/.style={nodes={draw=none, font = \scriptsize}},
				column 1/.style={nodes={draw=none, text width = 3em, font = \scriptsize}}] {
					$i=1$ & $\bullet$ & & $\times$ & & & & & & \\
		  	  $i=2$ & & & & $\bullet$ & & $\times$ & & & \\
		  	  $i=3$ & & & & & $\bullet$ & & $\times$ & & \\
				  $i=4$ & & & & & & & $\bullet$ & & $\times$ \\
		  	  & $c_1$ & $c_2$ & $c_3$ &$c_4$ & $c_5$ & $c_6$ & $c_7$ & $c_8$ & $c_9$ \\
      };

			\node[highligh, fill=red!20!white, fit = (m-1-2) (m-1-4)] {}; 
			\node[highligh, fill=blue!20!white, fit = (m-2-5) (m-2-7)] {}; 
			\node[highligh, fill=green!20!white, fit = (m-3-6) (m-3-8)] {}; 
			\node[highligh, fill=violet!20!white, fit = (m-4-8) (m-4-10)] {}; 

			\node[highligh, fill=black!30!white, fit = (m-2-2) (m-2-4)] {}; 
			\node[highligh, fill=black!30!white, fit = (m-3-2) (m-3-5)] {}; 
			\node[highligh, fill=black!30!white, fit = (m-4-2) (m-4-7)] {}; 

			\node[highligh, fill=black!30!white, fit = (m-1-5) (m-1-10)] {}; 
			\node[highligh, fill=black!30!white, fit = (m-2-8) (m-2-10)] {}; 
			\node[highligh, fill=black!30!white, fit = (m-3-9) (m-3-10)] {}; 

    \end{tikzpicture}
		\caption{A grid-representation of division~$\PSdivision$ from~\Cref{ex:division} with
    the heads ($\bullet$) and tails ($\times$) of the primitives. Shaded boxes
    represent invalid value combinations of parameters $i$ and~$j$ of
    function~$T$. For readability, we omit the sets of interim
    candidates.}\label{fig:grid-representation}
	\end{figure}
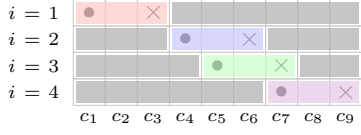

  \paragraph{Constructing Guesses and Running Time.} We run the DP procedure
  for each possible division $\PSdivision = (\rho, \PSdivsets, \PSdivfil)$. We
  first select some
  order~$\rho$ of candidates out of~$\candidatesnr!$ of them. Then, we guess the
  division size~$\PSdivisionSize$. Note that $\PSdivsets = (S_1, S_2, \ldots,
  S_\PSdivisionSize)$ can be represented as a vector $(s_1, s_2, \ldots,
  s_\PSdivisionSize)$, where $s_i$, $i \in [\PSdivisionSize]$, denotes the
  rightmost candidate of the interval of committee~$S_i$ given ordering~$\rho$.
  Hence, the guessed division size~$\PSdivisionSize$ varies from one to at
  most~$\candidatesnr$, where the upper bound is not tight and describes the
  maximum number of distinct equisized interval committees
  coming from taking
  $\PSdivisionSize=\candidatesnr-\commsize$ and $s_i = \commsize+i$, for $i
  \in [\PSdivisionSize]$.
  Importantly, $\tau$ does not contribute to the number
  of possible values of~$\PSdivisionSize$, as it is only the DP procedure that
  achieves $\tau$~committees for some~\PSdivision.
  Next, we guess each entry of the~\PSdivsets{} vector,
  except for the first one which is fixed. To this end, we choose the value of
  each $s_i$, $i \in \{2, 3, \ldots, \PSdivisionSize\}$, from range
  $\{s_{i-1}+1, s_{i-1}+2, \ldots, s_{i-1} + \commsize + 1\}$ (recall that all
  primitives in~\PSdivsets{} have to form an interval according~$\PSorder$).
  Finally, for each non-empty intersection $X$ of neighboring primitives, of
  which there are at most~$\PSdivisionSize$, we either select an empty interim
  set or take next unused~$\commsize - |X|$~candidates according to~$\PSorder$
  (as long as the latter is possible). This totals in at most
  $\candidatesnr!  \cdot \candidatesnr \cdot (\commsize+1)^{\candidatesnr} \cdot
  2^\candidatesnr$ guesses. Overall, the algorithm is clearly fixed-parameter
  tractable with respect to~$\candidatesnr$.
\end{proof}

The approach behind~\Cref{thm:fpt-m} in fact works for any quality measures of
committee series quality and committee scoring function as long as their values
are polynomially bounded by the input size. This leads to a general result
covering much more scoring functions than those described in~\Cref{sec:prelim}.
\begin{corollary}
  There is an $\fpt(m)$-algorithm with running time $\OO^\star(m!(k-1)^m)$ solving
  \pname{\alpha}{\beta} for arbitrary computable aggregation~$\alpha$
  and a committee scoring function~$\beta$ whose values are polynomially bounded
  in the instance size.
\end{corollary}

\section{The Case of Short Time Horizon}

The algorithms from the previous section are no longer of use
for large pools of candidates. In such cases, approaches tailored to small values of
parameter~\timelimit{}, which is a formal framing of short time horizon in our
model, might come to the rescue. Imagine, for example, an online
streaming platform building daily music recommendation for the following week,
based on a user's past-week activity. While such a task might likely include even
hundreds of songs, it is still only seven lists (committees) that we want to
arrange in a series.

In most cases that we study, however, we are not expecting \fpt{}-algorithms for
parameter~$\timelimit$. The
Chamberlin--Courant rule, in its variants for approval as well as ordinal
preferences, is \wtwo{}-hard for parameter ``committee size'' already in the
multiwinner voting
model~\citep{lu-bou:c:budgeted-social-choice,bet-sli-uhl:j:mon-cc-multivariate}.
Because these results are special cases of \pname{\alpha}{\beta} for both
committee series quality variants and $\beta \in \{\cc, \egalCCScoreA^\gamma,
\egalCCScoreO, \coverageScore\}$ for $\timelimit=1$, there is no hope for
fixed-parameter tractability for parameter~$\commsize + \timelimit$ in this
case.

Consequently, we focus on \fpt{}-algorithms with respect to $\timelimit$
assuming a constant committee size. On the one hand, such results are a clear
advancement from the theoretical perspective given that \pname{\utilOp}{\beta}
and \pname{\egal}{\beta} both are \nphard{} even for the constant committee size
for almost every committee scoring~$\beta$ that we
study~\citep{bre-kac-nie:c:successive-committees} (the exceptions are
\pname{\utilOp}{\WSFamily} and \pname{\utilOp}{\approvalScore}). On the other
hand, constant values of~$\commsize$ are still interesting from the practical
point of view as small committee sizes might appear in practice. For example, in
our toy example about daily music recommendation, providing too many songs each
day would be overwhelming for the user and ineffective in promoting featured
tracks.

Our first algorithm reinterprets \pname{\utilOp}{\beta} with $\candquota=1$
as an instance of~\wsetpacking{} and then applies a known procedure
of~\citet{DBLP:journals/siamdm/GoyalMPZ15} for solving the latter.

\begin{restatable}[\appsymb]{theorem}{utilfpttfone}
  There exists an algorithm that solves \pname{\utilOp}{\beta} in
  $\OO^\star(2.851^{(\commsize-0.5501)\timelimit})$~time for all studied~$\beta$ when
  the committee size $\commsize$ is constant and $\candquota=1$.
\label{thm:util-fpt-tf1}\end{restatable}

Next, we prove the result for constant $\candquota\geq 2$. Without loss of
generality, we assume that $\timelimit>\candquota$. Otherwise, we find a
committee of maximum score in polynomial time and repeat it $\candquota$~times.
Since $\commsize$
is a constant, computing a maximum score committee is polynomial-time solvable
for every committee scoring~$\beta$ that we study.

\begin{restatable}[\appsymb]{theorem}{utilfpttftwoormore}
  There exists an algorithm that solves \pname{\utilOp}{\beta} in
  $\OO^\star\left(2^{\commsize\timelimit(\candquota+1)}(2e)^{\commsize\timelimit}(\commsize\timelimit)^{\log
  (\commsize\timelimit)}\right)$ time for all studied~$\beta$ when the committee
  size $\commsize$ and $\candquota$ are constant.%
\label{thm:util-fpt-tf2}\end{restatable}
The proof of \Cref{thm:util-fpt-tf2} heavily depends on the following randomized
algorithm.

\begin{restatable}[\appsymb]{theorem}{randegalfptt}
   There exists a randomized algorithm that given an instance of
   \pname{\utilOp}{\beta} for all studied~$\beta$ either reports \failure{} or
   outputs a solution in $\OO^\star(2^{k\timelimit(f+1)}(2e)^{k\timelimit})$
   time. Moreover, if the algorithm is given a \yes-instance, it returns \yes{}
   with probability at least $\nicefrac{1}{2}$, and if the algorithm is given a
   \no-instance, it returns \no{} with probability 1.
\label{thm:rand-egal-fpt-t}\end{restatable}

\begin{proofsketch}
    A solution $\Co{S}=(S_1,\ldots,S_{\timelimit})$ is called \emph{\colorful}
    if every pair of candidates $s,s'$ in $\uplus_{i\in [\timelimit]}S_i$ has
    distinct colors. Now, our algorithm runs in three phases. In the first
    phase, we uniformly and independently at random color the candidates forming
    our solution. Then, we discard the sets that cannot be part of our solution
    under the coloring. In the third phase, we use dynamic programming to turn
    the coloring into a \colorful{} solution, if the coloring admits one. 
    \end{proofsketch}

We obtain~\Cref{thm:util-fpt-tf2} by derandomizing the algorithm
of~\Cref{thm:rand-egal-fpt-t}. To
this end we employ a \emph{$(p,q)$-perfect hash family}~\citep{alon1995color}
that we construct in the requested time applying the results
of~\citet{naor1995splitters} and~\citet{ParamAlgorithms15b}.

We provide more good news, by presenting an analogous result for the egalitarian
committee series evaluation. By this, we cover all of our cases of interest (for
constant~$\commsize$).

\begin{restatable}[\appsymb]{theorem}{egalfpttf}
  There exists an algorithm that solves \pname{\egal}{\beta} in
  $\OO^\star(2.851^{(k-0.5501)\timelimit})$ time for all studied~$\beta$ when
  the committee size $\commsize$ is constant.
\label{thm:egal-fpt-tf1}\end{restatable}

\section{Conclusions and Future Directions}
We provided the first algorithms for solving hard instances emerging from the model
of successive committee elections~\citep{bre-kac-nie:c:successive-committees}.
Extending the algorithmic understanding of a recent area of temporal elections,
our results concern potentially practical scenarios including small numbers of
candidates and a short time horizon. %
In the light of increasing resonance of experiments in computational social
choice~\citep{boe-fal-jan-kac---:c:guide},
our algorithms potentially
enable experimental
study on successive committee elections.

Theoretical follow-up directions include complementing our picture 
with parameterizations by ``number~$\votersnr$ of voters'' and ``committee
series quality~\vallowerbound.'' While both these parameters might turn out to
be quite large in practical applications, obtaining the respective results would
complete the picture of the parameterized complexity of the studied problems. We
note that parameter~$\vallowerbound$ in part inherits technical intricacies of
the study of the parametrization by (mis)representation for Chamberlin--Courant
rules~\citep{bet-sli-uhl:j:mon-cc-multivariate,che-roy:x:CC-intractability-misrepresentation}.
Importantly,
\citet[Theorem 5]{bre-kac-nie:c:successive-committees} implicitly excludes fixed-parameter tractability
of~\vallowerbound{} for
\pname{\egalOp}{\egalCCScoreA^1(W)}, while our results yield
it for
 \pname{\utilOp}{\cc} and \pname{\utilOp}{\coverageScore}
for~$\eta$.

On the more practical
side, our algorithms require an empirical treatment to verify its potential
applicability in practice. As our study only gives pessimistic running-time
upper bounds, a thorough investigation based on synthetic and real-life instances is
needed to check the algorithms behavior on more realistic input
instances.

\newpage

\section*{Acknowledgements}
We thank anonymous reviewers for their helpful comments. Pallavi Jain
acknowledges her support from SERB-SUPRA (grant number SPR/2021/000860) and IITJ
Seed Grant (grant no I/SEED/PJ/20210119). Andrzej Kaczmarczyk acknowledges
support from NSF CCF-2303372 and ONR N00014-23-1-2802.
This project
has received funding from the European Research Council (ERC) under the European
Union’s Horizon 2020 research and innovation programme (grant agreement No
101002854).

\begin{center}
  \includegraphics[width=3cm]{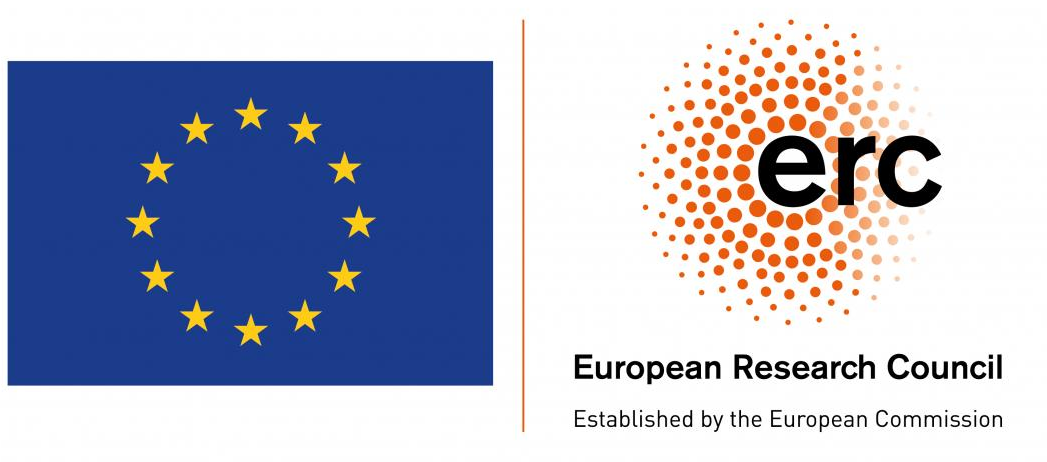}
\end{center}

\bibliographystyle{plainnat}
\bibliography{main}

\clearpage
\appendix

\begin{center}
\noindent
{\LARGE{\bf{}Appendix}}
\end{center}

\section{Few Candidates}

\divisionspace*
\begin{proof}
We show that if there is an optimal solution that does not belong to any space
of any division~$\PSdivision$, then this solution can be transformed to one that
belongs to some space like this.

First, let us consider a division~$\PSdivision = (\PSdivsets, \PSdivfil)$ such
that $\PSdivsets = (S_1, S_2, \ldots, S_\PSdivisionSize)$ and $\PSdivfil = (F_1,
F_2, \ldots, F_\PSdivisionSize)$ for some committee size~$\commsize$, and a time
horizon~$\timelimit$. Further, assume some order~$\rho$. The space of
$\PSdivision$ and~$\rho$ is formed by all combinations $(s_1, s_2, \ldots,
s_\PSdivisionSize, f_1, f_2, \ldots, f_\PSdivisionSize)$ such that $\sum_{i \in
[\PSdivisionSize]} (s_i + f_1) = \timelimit$ and $F_i = \emptyset \Rightarrow
f_i = 0$ for each $i \in [\PSdivisionSize]$. Each such combination forms a
committee series consisting of $s_1$ committees $S_1$ followed by $f_1$
committees $(S_1 \cap S_2) \cup F_1$ and so on up to $S_\PSdivisionSize$. Note
that each time committee $(S_i \cap S_{i+1}) \cup F_i$ would be too small, we
have~$f_i = 0$ by which we disallow such malformed committees.

From now on, let us assume some~$\candquota \leq 4$.  Let us assume that there
is some optimal committee series~$\committeesSeries$ such that there exist two
distinct committees~$S$ and~$S'$ such that~$S \cap S' \neq \emptyset$ and
between them there are at least two distinct committees~$H$ and~$H'$ that
contain candidates different than those in~$S \cup S'$. By consecutiveness
of~$\committeesSeries$, it holds that $(S \cap S') \subset H$ and~$(S \cap S')
\subset H'$ and that none of the candidates in these two committees is used
anywhere else then between (copies of) $S$ and~$S'$. By the frequency
requirement, it holds that candidates $H\setminus(S\cap S')$
and~$H'\setminus(S\cap S')$ are used fewer than~$\candquota$~times. By the same
requirement, we know that the total number of times~$H$ and~$H'$ are used is at
most~$\candquota-2$. Hence, from committees~$H$ and~$H'$ we can select the one,
say~$X$, with a better score (or any of them if the scores are equal), and
substitute the other one with~$X$. As a result, we use candidates~$X \setminus
(S\cap S')$ at most~$\candquota-2$~times in a consecutive subseries of
committees. Naturally, the consecutiveness and frequency of candidates in $S
\cap S'$ remains intact. The quality of the new committee is not worse then the
quality of the original committee. If it is the same, the new committee is also
optimal. It if it is better, we get a contradiction of optimality
of~$\committeesSeries$.

Except for the situation described above, it is clear that all other optimal
committees can be expressed by some order~$\rho$ and some
division~$\PSdivision$.

\end{proof}

\polymultfone*
\begin{proof}
We begin with the observation that for $f=1$, the problem reduces to finding a $k$-sized $t$ disjoint subsets of candidates that meets our score criteria. Towards this, we use the technique  of polynomial multiplication. 
 For this theorem, let the set of candidates be $C=\{1,\ldots,m\}$. Before we discuss our algorithm, let us introduce some relevant definitions. 
The {\em characteristic vector} of a subset $S\subseteq [m]$, denoted by $\chi(S)$, is an $m$-length binary string whose $i^{\textup{th}}$ bit is $1$ if and only if  $i \in S$.  Two binary strings of length $m$ are said to be \emph{disjoint} if for each $i\in [m]$, the $i^\textup{th}$ bits in the two strings are different. The {\em Hamming weight} of a binary string $S$, denoted by $\mathcal{H}(S)$, is the number of $1$s in the string $S$. A monomial $y^i$ is said to have Hamming weight $w$ if the degree term $i$, when represented as a binary string, has Hamming weight $w$. The {\em Hamming projection} of a polynomial $p(y)$ to $h$, denoted by $\mathcal{H}_{h}(p(y))$, is the sum of all the monomials of $p(y)$ which have Hamming weight $h$. %
The {\em representative polynomial} of $p(y)$, denoted by $\mathcal{R}(p(y))$, is a polynomial derived from $p(y)$ by changing the coefficients of all monomials contained in it to $1$. It is important to keep the small coefficients to achieve the desired running time.
 
We begin with the following known result.

\begin{restatable}[\citealp{DBLP:conf/ijcai/Gupta00T21,DBLP:journals/tcs/CyganP10}]{proposition}{HWDisjointness}
Subsets $S_{1}, S_{2} \subseteq W$ are disjoint if and only if Hamming weight of the monomial $x^{\chi(S_1)+\chi(S_2)}$ is $|S_{1}|+|S_{2}|$.
\label{prop:HW-disjointness}
\end{restatable}

In the algorithm, we compute a polynomial $p_{s,j}^\ell$, where $s\in [m], j\in [\timelimit]$, $\ell \in [\eta]$, such that a non-zero polynomial corresponds to the fact that there is a sequence $\Co{S}$ of  $j$ disjoint committees, each of size $k$, such that their union has $s$ candidates and $\utilOp(\beta(\Co{S}))\geq \ell$. We compute these polynomials recursively as follows. 

For $j=1$, $\ell\in [\eta]$ 
\begin{equation}\label{eq1:poly-mult}
p_{k,1}^\ell(y)=\sum_{\substack{Y\subseteq C, \\ |Y|=k, \beta(Y)\geq \ell}} y^{\chi(Y)} 
\end{equation}

For $j>1$ and $s\in [m]$, where $s$ is a multiple of $k$, $\ell \in [\eta]$ we compute as follows. 

\begin{equation}\label{eq2:poly-mult}
p_{s,j}^\ell(y)=\Co{R}\Big(\Co{H}_s\Big(\sum_{\substack{Y\subseteq C, |Y|=k, \\ \beta(Y)=\ell'(<\ell)}}y^{\chi(Y)}\cdot p_{(s-k),(j-1)}^{\ell-\ell'}(y)\Big)\Big)
\end{equation}

Towards the correctness, we argue that if the polynomial $p_{m,t}^\eta(y)$ is
non-zero, we obtained a desired sequence of committees, otherwise it is a \no-instance. In particular, we prove the following the result.

\begin{restatable}{lemma}{lem:polymult-correctness}
    The polynomial $p_{m,t}^\eta(y)$ is non-zero if and only if there exists a sequence $\Co{S}$ of $t$ disjoint committees, each of size $k$, such that their union has $s$ candidates and $\utilOp(\beta(\Co{S}))\geq \eta$.
\label{lem:polymult-correctness}\end{restatable}

\begin{proof}
    In the forward direction, suppose that $p_{m,t}^\eta(y)$ is a non-zero polynomial.  
    
    \begin{restatable}{claim}{clm:polymult-correctness-fwd}
    For  $j\in [\timelimit], s\in [m]$, where $s$ is a multiple of $k$, and $\ell \in [\eta]$, if  $p_{s,j}^\eta(\ell)$ is non-zero, then there exists a sequence $\Co{S}$ of $j$ disjoint committees, each of size $k$, such that their union has $s$ candidates and  $\utilOp(\beta(\Co{S}))\geq \ell$.        
    \label{clm:polymult-correctness-fwd}\end{restatable}
    
    \begin{proof}
        We prove it by induction on $j$.

        \noindent \emph{Base Case:} For $j=1$, we compute a polynomial for $s=k$. If $p_{k,1}^\ell(y)$ is non-zero, then due to construction of \Cref{eq1:poly-mult}, the claim is true. 

\noindent \emph{Inductive Step:} Suppose that that the clain is true for $j\leq i-1$. We prove it for $j=i$. Suppose that $p_{s,i}^\ell(y)$ is non-zero. Then, consider a monomial $y^{\chi(Y)}\cdot p_{(s-k),(j-1)}^{\ell-\ell'}$, where $\ell'=\beta(Y)$. Due to the induction hypothesis, there exists a sequence $\Co{S}$ of $j-1$ disjoint committees, each of size $k$, such that their union has $s-k$ candidates and  $\utilOp(\beta(S))\geq \ell-\ell'$. Let $Y'$ be the set of candidates that belong to a set in $\Co{S}$. Clearly, $|Y'|=s-k$. Since $y^{\chi(Y)}\cdot p_{(s-k),(j-1)}^{\ell-\ell'}$ is a monomial in $y_{s,j}^\ell$, due to \Cref{prop:HW-disjointness}, $Y$ and $Y'$ are disjoint. Thus, $\Co{S}'=\Co{S}\cup \{Y\}$ is a sequence of $j$ disjoint committees, each of size $k$, such that their union has $s$ candidates and the score is $\utilOp(\beta(\Co{S'}))+\beta(Y)\geq \ell$. 
    \end{proof}
    Due to \Cref{clm:polymult-correctness-fwd}, the proof of forward direction follows. For the reverse direction, we prove the following claim.

\begin{restatable}{claim}{clm:polymult-correctness-rev}
    For  $j\in [\timelimit], s\in [m]$, where $s$ is a multiple of $k$, and $\ell \in [\eta]$, if there exists a sequence $\Co{S}=\{S_1\uplus \ldots \uplus S_j\}$ of $j$ disjoint committees, each of size $k$, such that their union has $s$ candidates and  $\utilOp(\beta(\Co{S}))\geq \ell$, then the polynomial  $p_{s,j}^\eta(\ell)$ has a monomial $y^{\chi(S_1\uplus \ldots \uplus S_j)}$.        
    \label{clm:polymult-correctness-rev}\end{restatable}

    \begin{proof}
        We prove by induction on $j$.
        
     \noindent \emph{Base Case:} For $j=1$, suppose that there exists a committee $S\subseteq C$ of size $k$ such that $\beta(S)\geq \ell$. Then, due to \Cref{eq1:poly-mult}, the polynomial $p_{k,1}^\ell$ is non-zero. 

     \noindent \emph{Induction Step:} Suppose that that claim is true for $j=i-1$. We next prove it for $j=i$. Suppose that there exists a sequence  $\Co{S}=\{S_1,S_2,\ldots,S_{i}\}$ of $i$ disjoint committees, each of size $k$, such that their union has $s$ candidates and  $\utilOp(\beta(\Co{S}))\geq \ell$. Then, $\Co{S}'=\{S_1,\ldots,S_{i-1}\}$ is a sequence of $i-1$ disjoint committees, each of size $k$, such that their union has $s-k$ candidates and  $\utilOp(\beta(\Co{S}))\geq \ell - \beta(S_i)$. Thus, due to induction hypothesis, $p_{(s-k),(i-1)}^{\ell - \beta(S_i)}$ has a monomial $y^{\chi(S_1\uplus \ldots \uplus S_{i-1})}$. Since $S_1\uplus \ldots \uplus S_{i-1}$ is disjoint from $S_i$, the polynomial $y_{s,i}^\ell$ contains the monomial $y^{\chi(S_1\uplus \ldots \uplus S_{i})}$.
    \end{proof}
    This completes the proof.
\end{proof}
To argue the running time, we use the following known result.

\begin{restatable}{proposition}{prop:polynomial-multiplication}[\cite{moenck1976practical}]
    There exists an algorithm that multiplies two polynomials of degree $d$ in $\OO(d \log d)$ time.
\label{prop:polynomial-multiplication}\end{restatable}

Note that the degree of a polynomial in our algorithm is at most $2^m$ (when the
exponent has all 1s). Furthermore, every polynomial has at most $2^m\times \eta$
terms, where $\eta=(mn)^{\OO(1)}$

\paragraph{The Egalitarian Case.}
 Here, we compute a polynomial $p_{s,j}$, where $s\in [m], j\in [\timelimit]$, such that a non-zero polynomial corresponds to the fact that there is a sequence $\Co{S}$ of  $j$ disjoint committees, each of size $k$, such that their union has $s$ candidates and $\egal(\beta(\Co{S}))\geq \ell$. We compute these polynomials recursively as follows. 

For $j=1$

\begin{equation}\label{eq1:poly-mult-egal}
p_{k,1}(y)=\sum_{\substack{Y\subseteq C, \\ |Y|=k, \beta(Y)\geq \eta}} y^{\chi(Y)} 
\end{equation}

For $j>1$ and $s\in [m]$, where $s$ is a multiple of $k$,  we compute as follows. 

\begin{equation}\label{eq2:poly-mult-egal}
p_{s,j}(y)=\Co{R}\Big(\Co{H}_s\Big(\sum_{\substack{Y\subseteq C, |Y|=k, \\ \beta(Y)\geq \eta}}y^{\chi(Y)}\cdot p_{(s-k),(j-1)}(y)\Big)\Big)
\end{equation}
The proof of correctness is the same as for the utilitarian case.
\end{proof}

\utilfptmf*
\begin{proof}
	We generalize the approach from~\Cref{thm:pol-util-ftwo} by extending the
    function~$F$ with more parameters. Additional parameters allow us to keep
    track of the number of consecutive committees that candidates serve in.
    
    Specifically, we define a binary function~$F(A, j, s, G_1, \ldots, G_{f-1})$
    which has a similar meaning to that from the proof
    of~\Cref{thm:pol-util-ftwo}. Namely, it returns \true{}
    if among candidates~$A$, there is a
	series of quality at least~$s$ consisting of~$j$ committees,
    each of size $2$ such that
	the~$j$-th committee contains candidates
    from~$G_{x} \subseteq A$ in the $x$-th committee in the series; 
    in the opposite case, the function returns~\false{}.
    Our algorithm checks whether $F$ returns \true{} for at least one of~$F(C, \tau, \vallowerbound, G'_1, \ldots, G'_{f-1})$,
	for~$G'_x \subseteq C$, $x \in [\candquota-1]$, $|\bigcup_{x \in
    [\candquota-1]}G'_x| \leq k$.
    The algorithm returns \true{} if and only if the test succeeds.

	The algorithm computes the values of~$F$ applying the dynamic
    programming approach
	to the following recursive formula. We denote by~$\mathcal{W}(A, G_1,
	\ldots, G_{f-1}, k)$ a collection of all size-$\commsize$ sets~$W$ such that
	$W \subseteq A$ and $G_x \subseteq A$ for all~$x \in [\candquota-1]$
    \small
	\begin{align*}
	   &f(A, j, s, G_1, \ldots, G_{f-1})=\\
	   &\bigvee_{W \in \mathcal{W}(A, G_1,
	   \ldots, G_{f-1}, k)} f(A \setminus G_1, j-1, s-\score(W),\\
      &G_2, \ldots,
		G_{f-1},  W \setminus \bigcup_{x \in [f-1]} G_x).
     \end{align*}
     \normalsize
     The correctness is based on the very same reasoning as for the argument
     for the proof of~\Cref{thm:pol-util-ftwo}. For the running time,
     compared to that from~\Cref{thm:pol-util-ftwo}, here we
     only multiply the number of values of the function and the time required
     to compute a single value by additional terms~$4^\candquota$ as exactly
     $\candquota-2$~times, where the subtraction of two comes from the fact that
     the previous proof was given for $\candquota=2$.
\end{proof}

\section{Short Time Horizon}

\utilfpttfone*
\begin{proof}
  The result is due to a polynomial-time reduction from \pname{\utilOp}{\beta}
  with $f=1$ to \wsetpacking{}. In the latter, given a universe $U$, a family
  $\mathcal{F}$ of $d$-sized subsets of $U$, a weight function $w\colon
  \mathcal{F} \rightarrow \mathbb{R}_{\geq 0}$, and two integers $p,q$; the goal
  is to choose a set of $p$ disjoint sets of $\mathcal{F}$ whose total weight is
  at least $q$. \wsetpacking{} can be solved in
  $\OO^\star(2.851^{(d-0.5501)p})$ time~\citep{DBLP:journals/siamdm/GoyalMPZ15}.
  Given an instance $\Co{I}=(E=(C,V),k,f,\timelimit,\eta)$ of
  \pname{\utilOp}{\beta}, where $k$ is a constant and $f=1$, we construct an
  instance $\Co{J}=(U,\mathcal{F},w,p,q)$ as follows: $U=C$,
  $\mathcal{F}=\{S\subseteq C \colon |S|=k\}$, and the weight of every set $S\in
  \mathcal{F}$ is $\beta(S)$. Furthermore, $p=\timelimit,q=\eta.$ Next, we prove
  the correctness. 
\begin{lemma}
  For $\candquota = 1$, $\Co{I}$ is a \yes-instance of \pname{\utilOp}{\beta} if
  and only if $\Co{J}$ is a \yes-instance of \wsetpacking{}.
\end{lemma}

\begin{proof}
 Let $\Co{S}=(S_1,\ldots,S_{\timelimit})$ be a solution to $\Co{I}$. We claim
 that $\{S_1,\ldots,S_{\timelimit}\}$ is also a solution to $\Co{J}$.  Since
 $f=1$, for every $i,j\in [\timelimit]$, $S_i$ and $S_j$ are disjoint.
 Furthermore, size of each $S_i$ is $k$, thus for every $i\in [\timelimit]$,
 $S_i\in \Co{F}$.  Since $\utilOp(\beta(S))\geq \eta$,
 $\sum_{i\in[\timelimit]}w(S_i)\geq \eta$. Hence,
 $\{S_1,\ldots,S_{\timelimit}\}$ is a solution to $\Co{J}$. 

 In the reverse direction, let $\{S_1,\ldots,S_{\timelimit}\}$ be a solution to
 $\Co{J}$. Then, we claim that $\Co{S}=(S_1,\ldots,S_{\timelimit})$ is a
 solution to $\Co{I}$. Since for every $i,j\in [\timelimit]$, $S_i$ and $S_j$
 are disjoint, a candidate appears in only one committee. Since size of each set
 in $\Co{F}$ is $k$, each $S_i$ is of size $k$. Furthermore, due to the
 definition of weight functions, $\utilOp(\beta(S))\geq \eta$. 
\end{proof}
This completes the proof.
\end{proof}

\randegalfptt*

\begin{proof}
Consider an election~$E=(C,V)$ and an instance $\Co{I}=(E,k,f,\timelimit,\eta)$
    of \pname{\utilOp}{\beta}, where $k$ and $f$ are constant.
    We define a family $\Co{F}=\{S\subseteq C \colon |S|=k\}$ of all
    subsets of candidates of size $k$. Note that the goal is to choose $t$~sets
    from~$\Co{F}$ that represent committees and that lead to our solution (if it exists),
    i.e.,\ the selected committees overall achieve score $\eta$, and they can
    be ordered such that every candidate appears in at most $f$~consecutive
    ones.
    
    Our algorithm runs in three phases. In the first phase, we ``highlight'' the sets
    that could be part of our solution by coloring the candidates uniformly and
    independently at random. In the second phase, we discard the sets in $\Co{F}$
    that cannot be part of our solution based on our coloring. In the third phase,
    we find a ``colorful'' solution using dynamic programming, if it exists. 
    A solution $\Co{S}=(S_1,\ldots,S_{\timelimit})$ is called \emph{\colorful} if
    every pair of candidates $s,s'$ in $\uplus_{i\in [\timelimit]}S_i$ has distinct
    colors. Next, we describe our algorithm formally. If $\Co{I}$ is a \yes-instance,
    then let $\Co{S}=\{S_1,\ldots,S_{\timelimit}\}$ be a solution to $\Co{I}$. 

    \noindent{{\bf Phase I [Separation of candidates]}.}
     In this step, we  color the candidates uniformly and independently at random using $kt$ colors.
     The goal of this step is that ``with high probability'', we color the candidates in any of the set in $\Co{S}$ using distinct colors. Note that the number of candidates in $\uplus_{i=1}^t S_i$ is at most $kt$, where $k$ is the size of every committee and $t$ is the number of committees. Let $\chi \colon C \rightarrow [k\timelimit]$ be a \emph{coloring function} that colors each candidate in $C$ uniformly and independently at random.   
Due to \Cref{prop:success-prob}, the function $\chi$ colors the candidates in $\uplus_{i=1}^t S_i$ using distinct colors with probability at least $e^{-k\timelimit}$.  

\begin{proposition}{\rm \cite[Lemma 5.4]{ParamAlgorithms15b}}\label{prop:success-prob}
  Let $U$ be a universe and $X\subseteq U$. Let $\chi \colon U \rightarrow
  [|X|]$ be a function coloring each item of $U$ one of $|X|$ colors uniformly
  and independently at random. Then, the probability that the elements of $X$
  are colored with pairwise distinct colors is at least~$e^{-|X|}$.
\end{proposition}

    \noindent{\bf{Phase II [Filtration]}.} In the second phase of the algorithm, we delete a set from $\Co{F}$ that cannnot be part of the \colorful solution. In particular, we delete a set $S$ from $\Co{F}$ if for any pair of candidates $c,c'$ in $S$, $\chi(c)=\chi(c')$. 
    
    \noindent{\bf{Phase III [Dynamic Programming]}.} 

Given a coloring function $\chi$ and the family $\Co{F}$ obtained after filtration, we find a \colorful solution with respect to $\chi$ as follows. %
For $X\subseteq [k\timelimit], \ell \in [\eta]$, $T[1,X,\ell]=1$ if there exists
a set $S$ in $\Co{F}$ such that $\chi(S)=X$ and $\beta(X)\geq \ell$; otherwise
$0$.  For $i\in [\timelimit]$, where $i>f$, $X_1 \subseteq
[k\timelimit],\ldots,X_{f+1} \subseteq [k\timelimit]$, and $\ell_1\in
[\eta],\ldots,\ell_{f+1} \in [\eta]$,  we compute
$T[i,X_1,\ldots,X_{f+1},\ell_1,\ldots,\ell_{f+1}]$, where $X_1\cap X_{f+1} =
\emptyset$,  which has the following information:
$T[i,X_1,\ldots,X_{f+1},\ell_1,\ldots,\ell_{f+1}]=1$, if there exist $i$~sets in
$\Co{F}$ that can be ordered such that union of colors of first $i-f$ sets is in
$X_1$, and the score is $\ell_1$,   $(i-f+j)$th, $1\leq j \leq f$,    set has
color in $X_{j+1}$ and score $\ell_{j+1}$, and   a color appears at most $f$
times consecutively; otherwise $0$.

For $i=1, X\subseteq [k\timelimit], \ell\in [\eta]$, we compute $T[i,X,\ell]$
 as follows.
\begin{equation}\label{eq:paramt=basecase}
\begin{split}
  T[1,X,\ell]= & \begin{cases} 1 \text{ if there is a \colorful{} } S\in
    \Co{F} \text{ such}\\
    \;
     \text{that it has colors in } X,
      \text{and } \beta(S)=\ell \\
     0 \text{ otherwise}
     \end{cases} 
\end{split}
\end{equation}

For $i\in [\timelimit]$, where $i>f$, $X_1 \subseteq [k\timelimit],\ldots,X_{f+1} \subseteq [k\timelimit]$, such that $X_1\cap X_f = \emptyset$, and $\ell_1\in [\eta],\ldots,\ell_{f+1} \in [\eta]$, we compute as follows. When $i=f+1$, 
\begin{equation}\label{eq:paramt-f+1}
   \begin{split}
T[i,X_1,\ldots,X_{f+1},\ell_1,\ldots,\ell_{f+1}]   = & \underset{j\in [i]}{\bigwedge} T[1,X_j,\ell_j]
   \end{split}
\end{equation}
When $i>f+1$, 
\begin{equation}\label{eq:paramt-largei}
   \begin{split}
T[i,X_1,\ldots,X_{f+1},\ell_1,\ldots,\ell_{f+1}]   = & \\  \underset{\substack{X' \subseteq X_1, \\ \ell'\in [\ell_1] }}{\bigvee} \big(T[i-1,X_1\setminus X_{f}, X',  X_2,\ldots,X_{f}, \\ \ell_1-\ell',\ell',\ell_2,\ldots,\ell_{f}]\wedge T[1,X_{f+1},\ell_{f+1}]\big)
   \end{split}
\end{equation}

We return \yes{} if $T[t,X_1,\ldots,X_{f+1},\ell_1,\ldots,\ell_{f+1}]=1$ for
some $X_1 \subseteq [k\timelimit],\ldots,X_{f+1} \subseteq [k\timelimit]$, such
that $X_1\cap X_{f+1} = \emptyset$, and $\ell_1\in [\eta],\ldots,\ell_{f+1} \in
[\eta]$, where $\ell_1+\ldots+\ell_f\geq \eta$; otherwise we report \failure{}. 

Next, we prove the correctness of the algorithm. We begin with the following result.

    \begin{restatable}{lemma}{lem:prob-colorful}
       Given an instance $\Co{I}$ of \pname{\utilOp}{\beta} and a coloring
       function $\chi$, the above dynamic programming algorithm returns \yes{} iff there exists a \colorful solution to $\Co{I}$. 
    \label{lem:prob-colorful}\end{restatable}

\begin{proof}
We argue that \Cref{eq:paramt-f+1} and \Cref{eq:paramt-largei} correctly
compute $T[i,X_1,\ldots,X_{f+1},\ell_1,\ldots,\ell_{f+1}]$, where $X_1\cap
X_{f+1} = \emptyset$, for all $i\in [\timelimit]$, where $i>f$, $X_1
\subseteq [k\timelimit],\ldots,X_{f+1} \subseteq [k\timelimit]$, and $\ell_1\in
[\eta],\ldots,\ell_{f+1} \in [\eta]$. We prove it by induction on $i$. %
that $\timelimit>f$. 
    \begin{sloppypar}
    \emph{Base Case:} $i =f+1$. Note that $T[1,X,\ell]=1$ iff there exists a \colorful set in $\Co{F}$ whose $\beta$-score is $\ell$. Thus, $T[1,X,\ell]$ is correct. Thus, due to \Cref{eq:paramt-f+1}, $T[f+1,X_1,\ldots,X_{f+1},\ell_1,\ldots, \ell_{f+1}]=1$ iff there exists sets $S_1,\ldots,S_{f+1}$ such that $\chi(S_j)=X_j$ and $\beta(S_j)=\ell_j$, where $j\in [f+1]$. Thus, \Cref{eq:paramt-f+1} computes $T[f+1,X_1,\ldots,X_{f+1},\ell_1,\ldots, \ell_{f+1}]=1$ correctly. 
\end{sloppypar}
    \emph{Induction Step:} For $i>f+1$, $T[i,X_1,\ldots,X_{f+1},\ell_1,\ldots, \ell_{f+1}]=1$ is computed using \Cref{eq:paramt-largei}.  Suppose that that the claim is true for $i\leq j-1$, for all $X_1\subseteq [kt],\ldots, X_{f+1}\subseteq [kt]$ such that $X_1\cap X_{f+1}=\emptyset$, and $\ell_1\in [\eta],\ldots, \ell_{f+1}\in [\eta]$. We prove it for $i=j$. Towards this, we show the correctness of \Cref{eq:paramt-largei}. Suppose that $T[i,X_1,\ldots,X_{f+1},\ell_1,\ldots, \ell_{f+1}]=1$. Then, there exists $i$ \colorful sets in $\Co{F}$, say $S_1,\ldots,S_i$ such that $\chi(S_1\cup \ldots \cup S_{i-f})=X_1$, $f(S_{i-f+j})=X_{j+1}$, where $j\in [f]$, and $\utilOp(\beta((S_1,\ldots,S_{i-f})))\geq \ell_1$, $\beta(S_{i-f+j})\geq \ell_j$, where $j\in [f]$. Furthermore, a color appears at most $f$ times consecutively. Thus, $\chi(S_{i-1})\cap \chi(S_1\cup \ldots \cup S_{i-f-1})=\emptyset$. Hence, $\chi(S_1\cup \ldots \cup S_{i-f-1})=X_1\setminus X_f$ as $\chi(S_{i-1})=X_f$. Consider $X'=\chi(S_{i-f})$. Then, $(S_1,\ldots,S_{i-1})$ is a feasible solution for $T[i-1,X_1\setminus X_f,\chi(S_{i-f}),X_2,\ldots,X_f, \ell_1-\beta(S_{i-f}),\beta(S_{i-f}),\ell_2,\ldots,\ell_f]$.  %
    Furthermore $T[1,X_{f+1},\ell_{f+1}]=1$, %
    RHS of \Cref{eq:paramt-largei} is $1$. In the reverse direction, suppose that there exists $X'\subseteq X_1, \ell' \in [\ell_1]$ such that $T[i-1,X_1\setminus X_{f}, X',  X_2,\ldots,X_{f}, \\ \ell_1-\ell',\ell',\ell_2,\ldots,\ell_{f}]\wedge T[1,X_f,\ell_{f+1}]=1$. Thus, there exists sets $S_1,\ldots,S_{i-1}$ such that $\chi(S_1\cup \ldots S_{i-1-f})=X_1\setminus X_f$, $\chi(S_{i-f})=X'$, and $\chi(S_{i-f+j})=X_{j+1}$, where $j\in [f]$,  every color appears at most $f$ times consecutively, and $\utilOp(\beta((S_1\cup \ldots S_{i-1-f})))\geq \ell_1-\ell'$, $\beta(S_{i-f})=\ell'$, and $\beta(S_{i-f+j})=\ell_{j+1}$. Furthermore, there exists a set $S_i$ such that $\chi(S_i)=X_{f+1}$ and $\beta(S_i)=\ell_{f+1}$. Consider $\Co{S}=(S_1,\ldots,S_{i-1},S_i)$. Since $X_1 \cap X_{f+1} = \emptyset$, $X'\subseteq X_1$, $\chi(S_i)\cap \chi(S_1\cup \ldots \cup S_{i-f})=\emptyset$. Furthermore, $\utilOp(\beta((S_1\cup \ldots S_{i-f})))\geq \ell_1$  Thus, $\Co{S}$ is a feasible solution for $T[i,X_1,\ldots,X_f,\ell_1,\ldots,\ell_f]$. Hence, this value is also $1$. Thus, \Cref{eq:paramt-largei} is correct, and the claim follows due to induction hypothesis. 
\end{proof}

\begin{restatable}{lemma}{lem:success-prob}
  Given a \yes-instance $\Co{I}$, the algorithm returns \yes{} with probability
  at least $1/2$, and if the algorithm is given a \no-instance, the algorithm
  returns \no{} with probability $1$.
\label{lem:success-prob}\end{restatable}

\begin{proof}
Suppose that $\Co{I}$ is a \yes-instance, and
$\Co{S}=(S_1,\ldots,S_{\timelimit})$ be one of its solution. Then, due to
\Cref{prop:success-prob}, $\chi$ colors  $\cup_{i\in [\timelimit]}S_i$
distinctly with probability at least $e^{-k\timelimit}$. Thus, due to
\Cref{lem:prob-colorful}, we return \yes{} with probability at least $e^{-k\timelimit}$.  Thus, to boost the success probability to a constant, we repeat the algorithm independently $e^{k\timelimit}$ times. Thus, the success probability is at least $$1-\Big(1-\frac{1}{e^{k\timelimit}}\Big)^{e^{k\timelimit}} \geq 1-\frac{1}{e} \geq \frac{1}{2}$$
    If the algorithm returns a solution, then $\Co{I}$ is clearly a \yes-instance. 
\end{proof}

    Note that we compute $\timelimit\cdot 2^{k\timelimit(f+1)}\cdot \eta^{f+1}$ entries in dynamic programming algorithm and each entry can be computed in $2^{k\timelimit} \eta$ time. Since we repeat the algorithm $e^{k\timelimit}$ times, the algorithm runs in $\OO(2^{k\timelimit(f+1)}(2e)^{k\timelimit}\eta^{f+2})$ time. 
\end{proof}

\utilfpttftwoormore*
\label{app:util-fpt-two-or-more}
\begin{proof}
We derandomize the algorithm of~\Cref{thm:rand-egal-fpt-t} using $(p,q)$-perfect
hash family to obtain a deterministic algorithm for our problem. Towards this,
we first define the notion of $(p,q)$-perfect hash family. 

\begin{definition}[$(p,q)$-perfect hash family]{\rm \cite{alon1995color}}
For non-negative integers $p$ and $q$, a family of functions $f_1,\ldots,f_t$
from a universe $U$ of size $p$ to a universe of size $q$ is called a
$(p,q)$-perfect hash family, if for any subset $S\subseteq U$ of size at most
$q$, there exists $i\in [t]$ such that $f_i$ is injective on $S$. 
\end{definition}

According to the following result, we can construct a $(p,q)$-perfect hash
family.   

\begin{proposition}[\cite{naor1995splitters,ParamAlgorithms15b}]\label{prop:hash family construction}
There is an algorithm that given $p,q \geq 1$ constructs a $(p,q)$-perfect hash
family of size $e^qq^{\OO(\log q)}\log p$ in time $e^qq^{\OO(\log q)}p \log p$.
\end{proposition}

Instead of taking a random coloring $\chi$, we construct an $(m,kt)$-perfect
hash family $\Co{U}$ using \Cref{prop:hash family construction}. Then, for each
function $g\in \Co{U}$, we invoke the algorithm in \Cref{thm:rand-egal-fpt-t}
with the coloring function $\chi=g$. If there exists a solution
$\Co{S}=(S_1,\ldots,S_\timelimit)$ to $\Co{I}$,  then there exists a function
$g\in \Co{U}$ that is injective on $\cup_{i\in \tau}S_i$, since $\Co{U}$ is an
$(m,k\timelimit)$-perfect hash family. Consequently, due to
\Cref{lem:prob-colorful}, the algorithm returns \yes{}. Hence, we obtain
\Cref{thm:util-fpt-tf2}.
\end{proof}

\egalfpttf*

\begin{proof}
  For the case of~$\candquota=1$, the proof is almost the same as that of
  \Cref{thm:util-fpt-tf1}. The only difference is that here we reduce the
  problem to \textsc{Set Packing} and only keep the set $S$ in the family
  $\Co{F}$ if $\beta(S)\geq \eta$.  The running time is due to the same
  algorithm as that of~\Cref{thm:util-fpt-tf1}.

  By~\Cref{prop:egal2-egal1}, the results for~$\candquota=1$ can be generalized
  to any positive value of~$\candquota$.
\end{proof}

\end{document}